\documentclass[11pt]{article}
\usepackage{amsmath,amsfonts,amsthm,amssymb,color}
\usepackage[margin=1.0in]{geometry}
\usepackage{float}
\usepackage{fancybox}
\usepackage{hyperref}
\usepackage{cleveref}
\usepackage{subcaption}
\usepackage[dvipsnames]{xcolor}

\newtheorem{theorem}{Theorem}[section]

\newtheorem{proposition}[theorem]{Proposition}

\newtheorem{fact}[theorem]{Fact}

\newtheorem{question}[theorem]{Question}
\usepackage{thmtools}
\usepackage{thm-restate}
\newtheorem{lemma}[theorem]{Lemma}

\theoremstyle{definition}
\newtheorem{Definition}[theorem]{Definition}

\usepackage[tikz]{bclogo}
\usepackage{circuitikz}
\usepackage{tikz}

\usepackage{comment}

\hypersetup{
  colorlinks=true,
  citecolor=blue, 
  linkcolor=blue, 
  urlcolor=blue, 
}
\crefformat{equation}{(#2#1#3)}

\newenvironment{fminipage}%
  {\begin{Sbox}\begin{minipage}}%
  {\end{minipage}\end{Sbox}\fbox{\TheSbox}}

\def\prob#1#2{\mbox{\normalfont{Pr}}_{#1}\left[ \normalfont{#2} \right]}

\def\defeq{\stackrel{\mathrm{def}}{=}}

\def\abs#1{\left|#1  \right|}

\newcommand\ppi{\boldsymbol{\pi}}
\newcommand\cchi{\boldsymbol{\chi}}

\newcommand\uu{\boldsymbol{\mathit{u}}}
\newcommand\vv{\boldsymbol{\mathit{v}}}
\newcommand\ww{\boldsymbol{\mathit{w}}}

\newcommand\xx{\boldsymbol{\mathit{x}}}

\newcommand\LL{\boldsymbol{\mathit{L}}}

\newcommand\WW{\boldsymbol{\mathit{W}}}

\newcommand\sink{v_\textnormal{sink}}
\newcommand\Path[1]{\textsc{Path}_{#1}}
\renewcommand\line{\textsc{Line}}
\renewcommand\square[1]{\textsc{Square}_{#1}}
\newcommand\cube[1]{\text{$d$-}\textsc{Cube}_{#1}}

\renewcommand{\deg}[1]{\text{deg}(#1)}
\newcommand{\er}{\mathcal{R}_{\textnormal{eff}}}

\renewcommand\ww{{\mathit{w}}}
\renewcommand\WW{{\mathit{w}}}
\newcommand{\miP}[1]{{\min_{\leq #1}({\WW})}}
\newcommand{\maP}[1]{{\max_{\leq #1}({\WW})}}

\newcommand{\distr}[1]{\WW \sim \mathcal{W}^{\line}(i)}
\newcommand{\Z}{\mathbb{Z}}
\newcommand{\R}{\mathbb{R}}
\newcommand{\tcl}{\textnormal{tcl}}

\makeatletter
\let\@@pmod\pmod
\DeclareRobustCommand{\pmod}{\@ifstar\@pmods\@@pmod}
\def\@pmods#1{\mkern4mu({\operator@font mod}\mkern 6mu#1)}
\makeatother

\begin{document}

\title{
Nearly Tight Bounds for Sandpile Transience on the Grid
}

\author{
David Durfee
\\
Georgia Institute of Technology
\\
\href{mailto:ddurfee@gatech.edu}{\texttt{ddurfee@gatech.edu}}
\and
Matthew Fahrbach\thanks{
  Supported in part by a National Science Foundation Graduate Research
  Fellowship under grant DGE-1650044.
}
\\
Georgia Institute of Technology
\\
\href{mailto:matthew.fahrbach@gatech.edu}{\texttt{matthew.fahrbach@gatech.edu}}
\and
Yu Gao\thanks{
  A substantial portion of this work was completed while the author visited the
  Institute for Theoretical Computer Science at Shanghai University of Finance
  and Economics.
}
\\
Georgia Institute of Technology
\\
\href{mailto:ygao380@gatech.edu}{\texttt{ygao380@gatech.edu}}
\and
Tao Xiao\footnotemark[2]
\\
Shanghai Jiao Tong University
\\
\href{mailto:xt\_1992@sjtu.edu.cn}{\texttt{xt\_1992@sjtu.edu.cn}}
}

\date{}
\maketitle


\begin{abstract}
We use techniques from the theory of electrical networks to give nearly tight
bounds for the transience class of the Abelian sandpile model on the
two-dimensional grid up to polylogarithmic factors.
The Abelian sandpile model is a discrete process on graphs that is
intimately related to the phenomenon of self-organized criticality.
In this process, vertices receive grains of sand, and once the number of grains
exceeds their degree, they topple by sending grains to their neighbors. 
The transience class of a model is the maximum number of grains
that can be added to the system before it necessarily reaches its steady-state
behavior or, equivalently, a recurrent state.
Through a more refined and global analysis of electrical potentials and random
walks,
we give an $O(n^4\log^4{n})$ upper bound and an $\Omega(n^4)$ lower bound
for the transience class of the $n \times n$ grid.
Our methods naturally extend to $n^d$-sized $d$-dimensional grids to give 
$O(n^{3d - 2}\log^{d+2}{n})$ upper bounds and
$\Omega(n^{3d -2})$ lower bounds.
\end{abstract}

\pagenumbering{gobble}
\newpage

\pagenumbering{arabic}
\section{Introduction}
\label{sec:intro}
The Abelian sandpile model is the canonical dynamical system used to
study \emph{self-organized criticality}.
In their seminal paper, Bak, Tang, and Wiesenfeld~\cite{BakTW87} proposed the
idea of self-organized criticality to explain several ubiquitous patterns in
nature typically viewed as complex phenomena, such as catastrophic events
occurring without any triggering mechanism, the fractal behavior of mountain
landscapes and coastal lines, and the presence of pink noise in electrical
networks and stellar luminosity.
Since their discovery, self-organized criticality has been observed in an
abundance of disparate scientific fields~\cite{Bak96, WatkinsPCCJ16}, including
condensed matter theory~\cite{WijngaardenWAM06},
economics~\cite{BiondoPR15, ScheinkmanW94},
epidemiology~\cite{SabaMM14},
evolutionary biology~\cite{Phillips14},
high-energy astrophysics~\cite{Aschwanden11, Mineshige94},
materials science~\cite{RamosAM09},
neuroscience~\cite{BrochiniAAASK16, LevinaHG07}, 
statistical physics~\cite{Dhar06, Manna91},
seismology~\cite{SornetteS89}, and
sociology~\cite{KronG09}.
A stochastic process is a self-organized critical system if it naturally evolves
to highly imbalanced critical states where slight local disturbances can 
completely alter the current state.
For example, when pouring grains of sand onto a table,
the pile initially grows in a predictable way, but as it becomes steeper and
more unstable, dropping a single grain can spontaneously cause an avalanche
that affects the entire pile.
Self-organized criticality differs from the critical point of a phase
transition in statistical physics, because a self-organizing system does not
rely on tuning an external parameter. Instead, it is insensitive to all
parameters of the model and simply requires time to reach criticality, which is
known as the transient period.
Natural events empirically operate at a critical point between order and chaos,
thus justifying our study of self-organized criticality.

Dhar~\cite{Dhar90} developed the \emph{Abelian sandpile model} on finite directed graphs with
a sink vertex to further understand self-organized criticality.
The Abelian sandpile model, also known as a chip-firing game~\cite{BjornerLS91chip},
on a graph with a sink is defined as follows.
In each iteration a grain of sand is added to a non-sink vertex of the graph.
While any non-sink vertex $v$ contains at least $\deg{v}$ grains of sand, a
grain is transferred from $v$ to each of its neighbors.
This is known as a \emph{toppling}.
When no vertex can be toppled, the state is \emph{stable} and the iteration
ends.
The sink absorbs and destroys grains, and
the presence of a sink guarantees that every toppling procedure eventually
stabilizes.
An important property of the Abelian sandpile model is that the order
in which vertices topple does not affect the stable state.
Therefore, as the process evolves it produces a sequence of stable states.
From the theory of Markov chains, 
we say that a stable state is \emph{recurrent} if it can be revisited;
otherwise it is \emph{transient}.

In the self-organized critical state of the Abelian sandpile model on a graph
with a sink, 
transient states have zero probability and recurrent states occur
with equal probability~\cite{Dhar90}.
As a result, recurrent configurations model the
steady-state behavior of the system.
Thus, the natural algorithmic question to ask about self-organized
criticality for the Abelian sandpile model is:
\begin{question}
  How long in the worst case does it take for the process to reach its
  steady-state behavior or, equivalently, a recurrent state?
\end{question}
\noindent
Starting with an empty configuration,
if the vertex that receives the grain of sand is chosen uniformly at
random in each step, Babai and Gorodezky~\cite{BabaiG07}
give a simple solution that is polynomial in the number of edges of the graph
using a coupon collector argument.
In the worst case, however,
an adversary can choose where to place the grain of sand in each iteration.
Babai and Gorodezky analyze the \emph{transience class} of the model to
understand its worst-case behavior, which is defined as the maximum number of
grains that can be added to the empty configuration before the configuration
necessarily becomes recurrent.
An upper bound for the transience class of a model is an upper bound for the
time needed to enter self-organized criticality.

\subsection{Results}
We give the first nearly tight bounds (up to polylogarithmic factors)
for the transience class of the Abelian sandpile model on the
$n \times n$ grid with all boundary vertices connected to the sink.
This model was first studied in depth by
Dhar, Ruelle, Sen, and Verma~\cite{DharRSV95}, and it has since been the
most extensively studied Abelian sandpile model due to its role in
algebraic graph theory, theoretical computer science, and statistical physics.
Babai and Gorodezky~\cite{BabaiG07} initially established that the transience
class of the grid is polynomially bounded by $O(n^{30})$, which was unexpected
because there are graphs akin to the grid with exponential transience classes.
Choure and Vishwanathan~\cite{ChoureV12} improved the upper bound for the
transience class of the grid to $O(n^7)$ and gave a lower bound
of~$\Omega(n^{3})$ by viewing the graph as an electrical network and relating
the Abelian sandpile model to random walks on the underlying graph.
Moreover, they conjectured that the transience class of the grid is $O(n^4)$,
which we answer nearly affirmatively.

\begin{restatable}[]{theorem}{mainGrid}\label{thm:mainGrid}
The transience class of the Abelian sandpile model on the $n \times n$ grid is
$O(n^{4} \log^4{n})$.
\end{restatable}
\begin{restatable}[]{theorem}{mainLowerGrid}\label{thm:mainLowerGrid}
The transience class of the Abelian sandpile model on the $n \times n$ grid is
$\Omega(n^{4})$.
\end{restatable}

\noindent
Our results establish how fast the system reaches its steady-state behavior in
the adversarial case, and they corroborate empirical observations about
natural processes exhibiting self-organized criticality.
Our analysis directly generalizes to higher-dimensional cases, giving the
following result.


\begin{restatable}[]{theorem}{DimD}\label{thm:DimD}
  For any integer $d \ge 2$, the transience class of the Abelian sandpile model on the~$n^d$-sized
  $d$-dimensional grid is $O(n^{3d - 2}\log^{d+2}{n})$
  and $\Omega(n^{3d -2})$.
\end{restatable}

In addition to addressing the main open problem in \cite{BabaiG07} and
\cite{ChoureV12}, we begin to shed light on Babai and Gorodezky's inquiry
about sequences of graphs that exhibit polynomially bounded transience classes.
Specifically, for hypergrids (a family of locally finite graphs with high
symmetry) we quantify how the transience class grows as a function of the size
and local degree of the graph.
When viewed through the lens of graph connectivity,
such transience class bounds are surprising 
because grids have low algebraic connectivity, yet
we are able to make global structural arguments using only the fact that
grids have low maximum effective resistance when viewed as electrical networks.
By doing this, we avoid spectral analysis of the grid and
evade the main obstacle in Choure and Vishwanathan's analysis.
Our techniques suggest that low effective resistance 
captures a different but similar phenomenon to high conductance and high
edge expansion for stochastic processes on graphs.
This distinction between the role of a graph's effective resistance and
conductance could be an important step forward for 
building a theory for discrete diffusion processes analogous to the
mixing time of Markov chains.
We also believe our results have close connections to randomized, distributed
optimization algorithms for flow problems~\cite{BecchettiBDKM13,BonifaciMV12,Mehlhorn13,StraszakV16:arxiv,StraszakV16a,StraszakV16b}, where the
dynamics of self-adjusting sandpiles (a Physarum slime mold in their model) are
governed by electrical flows and resistances.

\subsection{Techniques}

Our approach is motivated by the method of Choure and Vishwanathan~\cite{ChoureV12}
for bounding the transience class of the Abelian sandpile model on graphs using
electrical potential theory and the analysis of random walks.
Viewing the graph as an electrical network with a voltage source at some vertex
and a grounded sink, we give more accurate voltage estimates by carefully
considering the geometry of the grid.
We use several lines of symmetry to compare escape probabilities of random
walks with different initial positions, resulting in a new technique for
comparing vertex potentials.
These geometric arguments can likely be generalized to other lattice-based
graphs.
As a result, we get empirically tight inequalities for the sum of all vertex
potentials in the grid and the voltage drop between opposite corners of the
network.


For many of our voltage bounds, we interpret a vertex potential as an escape
probability and decouple the corresponding two-dimensional random walks on the
grid into independent one-dimensional random walks on a path graph.
Decoupling is the standout technique in this paper, because it allows us to
apply classical results about simple symmetric random walks on $\Z$ (such as
the reflection principle), which we extend as needed using conditional
probability arguments.
By reducing from two-dimensional random walks to one-dimensional
walks, we utilize standard probabilistic tools including Stirling's
approximation, Chernoff bounds, and the negative binomial distribution.
Since we consider many different kinds of events in our analysis,
\Cref{sec:path} is an extensive collection of probability inequalities
for symmetric $t$-step random walks on $\Z$ with various boundary conditions.
We noticed that some of these inequalities are directly
related to problems in enumerative combinatorics without closed-form
solutions~\cite{Everett77}.


Lastly, we leverage well-known results about effective resistances of the $n
\times n$ grid when viewed as an electrical network.
We follow Choure and Vishwanathan in using the potential reciprocity theorem
to swap the voltage source with any other non-sink vertex, but we use this
theorem repeatedly with the fact that the effective resistance between
any non-sink vertex and the sink is bounded between a constant and $O(\log n)$.
This approach enables us to analyze tractable one-dimensional random walk problems 
at the expense of polylogarithmic factors.

\section{Preliminaries}\label{sec:background}

\subsection{Abelian Sandpile Model}
Let $G = (V, E)$ be an undirected multigraph.
Throughout this paper all of the graphs we consider have a sink vertex
denoted by $\sink$.
The \emph{Abelian sandpile model} is a dynamical system on a graph $G$ used
to study the phenomenon of self-organized criticality.
A \emph{configuration} $\sigma$ on $G$ in the Abelian sandpile model
is a vector of nonnegative integers
indexed by the non-sink vertices such that $\sigma(v)$ denotes the
number of grains of sand on vertex $v$.
We say that a configuration is \emph{stable} if~$\sigma(v) < \deg{v}$
for all non-sink vertices and \emph{unstable} otherwise.
An unstable configuration $\sigma$ moves towards stabilization by
selecting a vertex $v$ such that $\sigma(v) \ge \deg{v}$ and
sending one grain of sand from $v$ to each of its neighboring
vertices. This event is called a \emph{toppling} of $v$, and it creates
a new configuration $\sigma'$ such that
$\sigma'(v) = \sigma(v) - \deg{v}$, $\sigma'(u) = \sigma(u) + 1$ for all
vertices~$u$ adjacent to $v$,
and $\sigma'(u) = \sigma(u)$ for all remaining vertices. 
This procedure eventually reaches a stable state because~$G$ has a sink.
Moreover, the order in which vertices topple does not affect the final stable
state.
The initial configuration of the Abelian sandpile model is typically
the zero vector, and in each iteration a grain of sand is placed at a vertex
(chosen either deterministically or uniformly at random).
The system evolves by stabilizing the configuration and then receiving another
grain of sand.

A stable configuration $\sigma$ is \emph{recurrent} if the process can eventually
return to $\sigma$.  Any state that is not recurrent is \emph{transient}.
Note that once the system enters a recurrent state, it can never visit a
transient state.
Babai and Gorodezky~\cite{BabaiG07} introduced the following notion to 
upper bound on the number of steps for the Abelian sandpile model to reach
self-organized criticality.


\begin{figure*}[ht]
\label{fig:transience-period}
\centering
\begin{subfigure}[t]{0.24\textwidth}
  \includegraphics[width=1.0\textwidth]{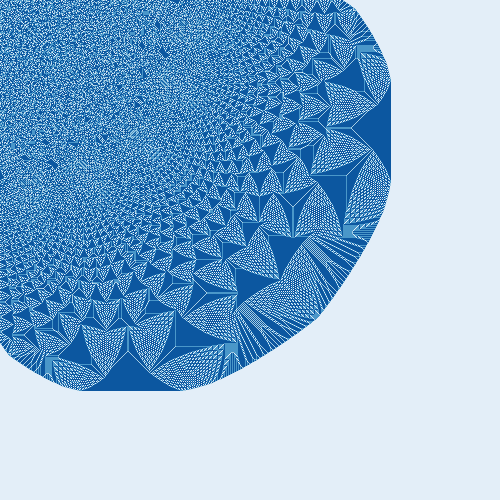}
  \caption{}
\end{subfigure}
\hfill
\begin{subfigure}[t]{0.24\textwidth}
  \includegraphics[width=1.0\textwidth]{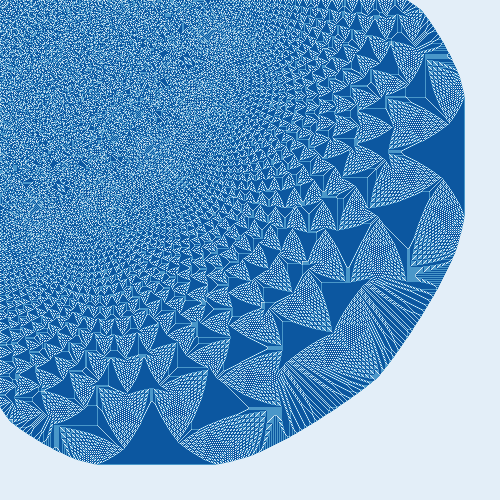}
  \caption{}
\end{subfigure}
\hfill
\begin{subfigure}[t]{0.24\textwidth}
  \includegraphics[width=1.0\textwidth]{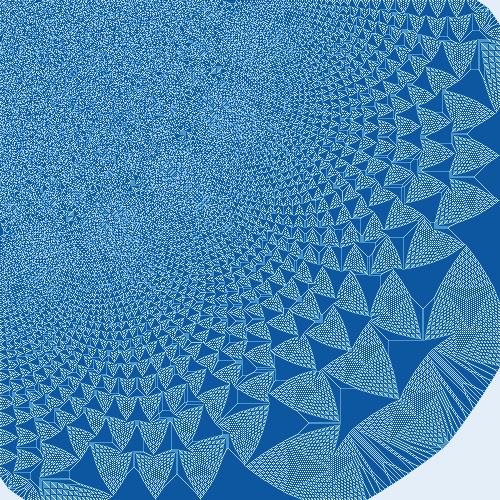}
  \caption{}
\end{subfigure}
\hfill
\begin{subfigure}[t]{0.24\textwidth}
  \includegraphics[width=1.0\textwidth]{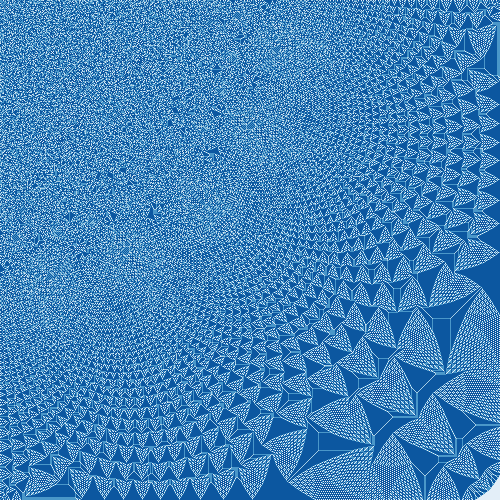}
  \caption{}
\end{subfigure}
\caption{Configurations of the Abelian sandpile model on the
  $500 \times 500$ grid during its transience period after placing
  (a) $10^{10}$ (b) $2\cdot10^{10}$ (c) $4\cdot10^{10}$ (d) $8\cdot10^{10}$
  grains of sand at $(1, 1)$.}
  \label{fig:transience-period}
\end{figure*}

\begin{Definition}
  The \emph{transience class} of the Abelian sandpile model of $G$ is
  the maximum number of grains that can be added to the empty configuration
  before the configuration necessarily becomes recurrent.
  We denote this quantity by $\tcl(G)$.
\end{Definition}

\noindent
In \Cref{fig:transience-period} we illustrate the transient configurations in the
\emph{transient period} of the Abelian sandpile model as it advances towards
its critical state.
We specifically show in this paper that by repeatedly placing grains of sand in
the top-left corner of the grid, we maximize the length of the transience
period up to a polylogarithmic factor.

In earlier related works,
Bj{\"o}rner, Lov{\'a}sz, and Shor~\cite{BjornerLS91chip}
studied a variant of this process without a sink and characterized the
conditions needed for stabilization to terminate.
They also related the spectrum of the underlying graph to the rate at which
the system converges.
In the model we study, an observation by Dhar~\cite{Dhar90} and Kirchoff's
theorem show that the stable recurrent states of the system are in bijection
with the spanning trees of~$G$.
Choure and Vishwanathan~\cite{ChoureV12} show that if every vertex in a
configuration has toppled then the configuration is necessarily recurrent,
which we use to bound the transience class.
The Abelian sandpile model also has broad applications to algorithms
and statistical physics,
including a direct relation to the $q$-state Potts model and Markov chain Monte
Carlo algorithms for sampling random spanning trees~\cite{Bhakta17, Dhar90, JerisonLP15,
Ramachandran17, Wilson10}.
For a comprehensive survey on the Abelian sandpile model, see~\cite{HolroydLMPPW08}.

\subsection{Random Walks on Graphs}

A walk~$\WW$ on $G$ is a sequence of vertices
$\WW^{(0)}, \WW^{(1)}, \dots, \WW^{(t_{\max})}$ such that every
$\WW^{(t+1)}$ is a neighbor of~$\WW^{(t)}$.
We let $t_{\max}=|w|$ denote the length of the walk.
A random walk is a process that begins at vertex $w^{(0)}$, and at each
time step $t$ transitions from $w^{(t)}$ to $w^{(t+1)}$ such that $w^{(t+1)}$ is
chosen uniformly at random from the neighbors of $w^{(t)}$.
Note that this definition naturally captures the effect of walking on a multigraph.
We consider walks that continue until reaching a set of sink vertices.
It will be convenient for our analysis to formally define these following families of walks.
\begin{Definition}
For any set of starting vertices $S$ and terminating vertices $T$ in the graph $G$, let
\begin{align*}
  \mathcal{W}\left(S \rightarrow T\right) \defeq
  \Big\{
    \WW : \WW^{(0)} \in S
    \text{, }
    \WW^{(i)} \not\in T\cup\{\sink\} \text{ for $0 \le i \le |\WW|-1$}
    \text{, and }
    \WW^{(|\WW|)} \in T
  \Big\}
\end{align*}
be the set of finite walks from $S$ to $T$.
\end{Definition}

\noindent
Observe that with this definition, walks $w$ of length~$0$ are permissible if
we have $w^{(0)} \in S \cap T$.
Throughout the paper it will be convenient to consider random walks from one
vertex $u$ to another vertex $v$ or the pair $\{v, \sink\}$.
We denote these cases by the notation
$
  \mathcal{W}\left(u \rightarrow v\right) = \mathcal{W}(\{u\} \rightarrow \{v\}).
$
If walks on multiple graphs are being considered, we use $\mathcal{W}^G(u\rightarrow v)$ to
denote the underlying graph.
Lastly, we consider the set of nonterminating walks in our analysis, so it
will be useful to define
\begin{align*}
  \mathcal{W}\left(S\right) \defeq
  \Bigg\{\WW \in \prod_{i=0}^\infty V: \WW^{(0)} \in S \text{ and }
  \WW^{(i)} \ne \sink \text{ for any $i \ge 0$}
  \Bigg\},
\end{align*}
which is the set of infinite walks from $S$. An analogous definition follows when
$S = \{u\}$. 



The focus of our study is the $n \times n$ grid graph, denoted by $\square{n}$.
Similar to previous works, we do not follow
the usual graph-theoretic convention of using $n$ to denote vertex count.
We formally define the one-dimensional projection of $\square{n}$ to be $\Path{n}$,
which has the vertex set
$\{1, 2, \dots, n\} \cup \{\sink\}$ and edges between $i$ and $i + 1$ for every
$1 \leq i \leq n  - 1$, as well as two edges connecting $\sink$ to $1$ and $n$.
Thus,~$\sink$ can be viewed as
$0$ and $n+1$.  If we remove the sink (which can be thought of as
letting $\sink = \pm\infty$) then the resulting graph is the one-dimensional line with
vertices $i \in \mathbb{Z}$ and edges between every pair $(i,i+1)$.
We denote this graph by $\line$
and
use the indices $i$, $j$, and $k$ to represent
its vertices.
Analyzing random walks on $\line$ is critical to 
our analysis, and it will be
useful to record the minimum and maximum position of $t$-step 
walks.

\begin{Definition}\label{def:1DminMax}
For an initial position $i \in \mathbb{Z}$ and walk
$\WW \in \mathcal{W}(i)$ on \line, let the $t$-step minimum and maximum
positions be
  \[\miP{t} \defeq \min_{0 \leq \widehat{t} \leq t} \ww^{( \widehat{t} )}\]
and
  \[\maP{t} \defeq \max_{0 \leq \widehat{t} \leq t} \ww^{( \widehat{t} )}.\]
\end{Definition}

We construct $\square{n}$ similarly. Its vertices are
$\left\{1, 2, \ldots, n \right\} \times \left\{1, 2, \ldots, n \right\} \cup \{\sink\}$,
and its edges connect any pair of vertices that differ in one coordinate.
Vertices on the boundary have edges connected to $\sink$ so that
every non-sink vertex has degree $4$.  With this definition of $\square{n}$,
each corner vertex has two edges to $\sink$ and non-corner vertices on the
boundary share one edge with $\sink$.
Since all vertices correspond to pairs of coordinates, we use the vector
notation $\uu = (\uu_{1}, \uu_{2})$ to denote coordinates on the grid,
as it easily extends to higher dimensions.
Throughout the paper, boldfaced variables denote vectors.
A $t$-step random walk on $\square{n}$ naturally induces a
$(t_{\max}+1) \times 2$ matrix.
We can decouple such a walk $\WW$ into its horizontal and vertical components,
using the notation $\WW_1$ for the change in position of the first
coordinate and $\WW_2$ for the change in position of the second
coordinate.
In general we use the notation $\WW_{\widehat{d}}$
to index into one of the dimensions $1 \leq \widehat{d} \leq d$
of a $d$-dimensional walk.
We do not record duplicate positions
when the walk takes a step in a dimension different than $\widehat{d}$,
so we have
$\abs{\WW} = \abs{\WW_1} + \abs{\WW_2} - 1$
when~$d=2$ since the initial vertex is present in both $\WW_1$ and $\WW_2$.

\subsection{Electrical Networks}

Vertex potentials are central to our analysis.
They have close connections with electrical voltages
and belong to the class of harmonic functions~\cite{DoyleS84}.
We analyze their relation to the transience class of general graphs.
For any non-sink vertex~$u$, we can define
a unique potential vector~$\ppi_u$ such that
  $\ppi_{u}(u) = 1$,
  $\ppi_{u}(\sink) = 0$,
  and for all other vertices $v \in V \setminus \{u, \sink\}$ we have
  \[
    \ppi_{u}(v) = \frac{1}{\deg{v}} \sum_{x\sim v} \ppi_{u}(x),
  \]
  where the sum is over the neighbors of $v$.
Thus, $\ppi_{u}(v)$ denotes the potential at $v$ when the boundary
conditions are set to $1$ at $u$ and~$0$ at the sink.
Since we analyze potential vectors in both $\Path{n}$
and $\square{n}$, we use superscripts to denote the graph
when context is unclear.


Choure and Vishwanathan showed that we can give upper and lower bounds on the
transience class using potentials, which we rephrase in the following theorem.
\begin{theorem}[{\cite{ChoureV12}}]
\label{thm:reduction}
If $G$ is a graph such that the degree of every non-sink vertex is bounded
by a constant, 
\begin{align*}
  &\tcl(G)
  = O\left( \max_{u,v \in V\setminus\{\sink\}} \left( \sum_{x \in V} \ppi_{u}(x) \right)
  \ppi_{u}(v)^{-1} \right)
\end{align*}
and
\begin{align*}
  \tcl(G)
  = \Omega\left( \max_{u,v \in V\setminus\{\sink\}} 
  \ppi_{u}(v)^{-1} \right).
\end{align*}
\end{theorem}

\noindent
All non-sink vertices have degree $4$, so we can apply Theorem~\ref{thm:reduction}
to $\square{n}$.

The following combinatorial interpretations of potentials as
random walks is fundamental to our investigation of the transience class of
$\square{n}$.
Note that we use boldfaced vector variables for non-sink vertices in $\square{n}$
as they can be identified by their coordinates.


\begin{fact}[{\cite{DoyleS84}}]
\label{fact:voltageRandomWalk}
For any graph $G$ and non-sink vertex $u$,
the potential $\ppi_{u}(v)$ is the probability of a random walk starting at $v$
and reaching $u$ before $\sink$.
\end{fact}

\begin{restatable}[]{lemma}{normalizingConstant}
\label{lem:normalizingConstant}
Let $\uu$ be a non-sink vertex of $\square{n}$. For any vertex $\vv$, we have
\[
  \ppi_{\uu}\left(\vv\right)
    = \sum_{\WW \in \mathcal{W}^{}\left(\vv \rightarrow \uu\right)} 4^{-|\WW|}.
\]
\end{restatable}
\noindent
We defer the proof of Lemma~\ref{lem:normalizingConstant}
to \Cref{sec:app_background}.

A systematic treatment of the connection between random walks and electrical
networks can be found in the monograph by Doyle and Snell~\cite{DoyleS84} or
the survey by Lov\'asz~\cite{Lovasz93}.  The following lemma is a key result for
our investigation, which states that a voltage source and a measurement point
can be swapped at the expense of a distortion in the potential equal to the ratio
of the effective resistances between the sink and the two vertices.  The
effective resistance between a pair of vertices $u$ and $v$, denoted as
$\er(u,v)$, can be formalized in several ways.  In the electrical
interpretation~\cite{DoyleS84}, effective resistance can be
viewed as the voltage needed to send one unit of current from $u$ to $v$ if
every edge in~$G$ is a unit resistor.
For a linear algebraic definition of effective resistance see
\cite{EllensSVJK11}.

  \begin{lemma}[{\cite[Potential Reciprocity]{ChoureV12}}]
\label{lem:potReciprocity}
Let $G$ be a graph (not necessarily degree-bounded) with sink $\sink$.
For any pair of vertices $u$ and $v$, we have
\[
  \er^{}\left(\sink, u\right) \ppi^{}_{u}(v)
  =
  \er^{}\left(\sink, v\right) \ppi^{}_{v}(u).
\]
\end{lemma}
This statement is particularly powerful for $\square{n}$, because
the effective resistance between any pair of vertices is bounded
between a constant and $O(\log n)$.
The following lemma makes use of a classical result that can be
obtained using Thompson's principle of the electrical flow~\cite{DoyleS84}.

\begin{restatable}[]{lemma}{boundedERGrid}
\label{lem:boundedERGrid}
For any non-sink vertex $\uu$ in $\square{n}$, 
\[
  1/4 
  \le \er^{} \left(\sink, \uu \right)
  \le 2 \log n + 1.
\]
\end{restatable}

\noindent
We give the proof of \Cref{lem:boundedERGrid} in \Cref{sec:app_background}.
When used together,
\Cref{lem:potReciprocity} and \Cref{lem:boundedERGrid} imply
the following result, which allows us to conveniently swap the source
vertex when computing potentials.


\begin{lemma}
\label{lem:swapSourceTarget}
For any non-sink vertices $\uu$ and $\vv$ in $\square{n}$, we have
  \[\ppi^{}_{\uu}\left(\vv\right)
  \leq \left(8\log n + 4\right) \ppi^{}_{\vv}\left(\uu\right).\]
\end{lemma}

Voltages and flows on electrical networks are central to
  many recent developments in algorithmic graph theory (e.g. modern maximum flow
  algorithms and interior point methods~\cite{ChristianoKMST11,Madry13}).
The convergence of many of these algorithms depend on the extremal voltage
values of the electrical flow that they construct.
As a result, we believe some of our techniques are relevant to the
grid-based instantiations of these algorithms.

\section{Upper Bounding the Transience Class}
\label{sec:grid}
In this section we prove the upper bound in Theorem~\ref{thm:mainGrid}
for the transience class of the Abelian sandpile model on the square grid.  Our
proof follows the framework of Choure and Vishwanathan in that we use
Theorem~\ref{thm:reduction} to reduce the proof to bounding the following two
quantities for any non-sink vertex $\uu \in V(\square{n})$:
\begin{itemize}
  \item We upper bound the potential sum 
  $
    \sum_{\vv \in V} \ppi^{}_{\uu}(\vv)
  $.
\item We lower bound the potential $\ppi^{}_{\uu}(\vv)$
    for all non-sink vertices $\vv$.
\end{itemize}

By symmetry we assume without loss of generality
that $\uu$ is in the top-left quadrant of
$\square{n}$ (i.e., we have $1 \leq \uu_{1}, \uu_{2} \leq \lceil n/2 \rceil$).
The principal idea is to use reciprocity from Lemma~\ref{lem:potReciprocity}
and effective resistance bounds from Lemma~\ref{lem:boundedERGrid} to
swap source vertices and bound $\ppi^{}_{\vv}(\uu)$ instead,
at the expense of a $O(\log{n})$ factor.
The second key idea is to interpret potentials as random
walks using Fact~\ref{fact:voltageRandomWalk} and then decouple two-dimensional
walks on $\square{n}$ into separate horizontal and vertical one-dimensional
walks on $\Path{n}$.
Using well-studied properties of
one-dimensional random walks, we achieve nearly tight bounds on 
$\tcl(\square{n})$.

We note that there is a natural trade-off in the choice of the source
vertex $\uu$.  Setting~$\uu$ near the boundary decreases 
vertex potentials because a random walk has a higher probability of escaping to
$\sink$ instead of $\uu$.
This improves the upper bound of the sum of vertex
potentials, but it weakens the lower bound of the minimum vertex potential.
For vertices $\uu$ that are not near the boundary, the opposite is true.
Therefore, we account for the choice of $\uu$ in our bounds.


\subsection{Upper Bounding the Potential Sum}\label{subsec:upper}

\begin{lemma}\label{lem:sumBound}
For any non-sink vertex $\uu$ in $\square{n}$, 
we have
\[
    \sum_{\vv \in V}
    \ppi^{}_{\uu}\left( \vv \right) 
      = O\left( \uu_1 \uu_2 \log^3{n} \right).
\]
\end{lemma}

\begin{proof}
We use Fact~\ref{fact:voltageRandomWalk} and \Cref{lem:normalizingConstant}
to interpret vertex potentials as random walks.
We can omit $\sink$ because any random walk starting there immediately
terminates.
By \Cref{lem:swapSourceTarget},
  \[\ppi^{}_{\uu}\left( \vv \right)
    = O\left(
      \ppi^{}_{\vv}\left( \uu \right)
      {\log n} \right),
  \]
so we apply the random walk interpretation
to potentials starting at $\uu$ instead of $\vv$.
Consider one such walk $\WW \in \mathcal{W}^{}(\uu \rightarrow \vv)$
and its one-dimensional decompositions $\WW_1$ and $\WW_2$.
The probability of a walk from $\uu$ reaching $\vv$ is equal to the
probability that two interleaved walks in $\Path{n}$ starting at $\uu_{1}$
and $\uu_{2}$ are present on $\vv_{1}$ and $\vv_{2}$, respectively,
at the same time
before either hits their one-dimensional sink $\sink = \{0,n+1\}$.

If we remove the restriction that these walks are present on $\vv_1$ and
$\vv_2$ at the same time and only require that they visit $\vv_1$ and $\vv_2$
before hitting $\sink$, then
each of these less restricted walks $\WW_d$ belongs to the class
$
 \mathcal{W}^{\Path{n}}\left( \uu_{d} \rightarrow \vv_{d} \right).
$
Viewing a walk $w$ on $\square{n}$ as infinite walk on the lattice $\Z^2$
induces independence between $w_1$ and $w_2$.
Thus, we obtain the upper bound
\begin{align*}
  \ppi^{}_{\vv}\left( \uu \right)
  &=
  \prob{w \sim \mathcal{W}^{\Z^2}(\uu)}{\text{$w$ hits $\vv$ before leaving $\square{n}$}}\\
  &\le 
  \text{Pr}_{w \sim \mathcal{W}^{\Z^2}(\uu)}[\text{$w_1$ hits $\vv_1$ before $\sink$ and
                      $w_2$ hits $\vv_2$ before $\sink$}]\\
  &=
  \prob{w \sim \mathcal{W}^{\Z^2}(\uu)}{\text{$w_1$ hits $\vv_1$ before $\sink$}}
  \cdot \prob{w \sim \mathcal{W}^{\Z^2}(\uu)}{\text{$w_2$ hits $\vv_2$ before $\sink$}}\\
  &=
  \ppi^{\Path{n}}_{\vv_1}(\uu_1) \cdot \ppi^{\Path{n}}_{\vv_2}(\uu_2).
\end{align*}

Summing over all choices of $\vv = (\vv_1, \vv_2)$ gives
\begin{align*}
  \sum_{\vv \in V} \ppi_{\vv}\left( \uu \right)
  \leq \left( \sum_{\vv_1 = 1}^n \ppi^{\Path{n}}_{\vv_{1}} \left(\uu_{1} \right) \right)
  \left( \sum_{\vv_2 = 1}^n \ppi^{\Path{n}}_{\vv_2} \left(\uu_{2} \right) \right).
\end{align*}
The potentials of vertices in $\Path{n}$ have the following closed-form
solution, as shown in \cite{DoyleS84}:
\[
  \ppi^{\Path{n}}_{\vv_1}\left(\uu_1\right)
=
\begin{cases}
  \frac{n + 1 - \uu_1}{n + 1 - \vv_1} & \qquad \text{if $\vv_1 \le \uu_1$},\\
  \frac{\uu_1}{\vv_1} & \qquad \text{if $\vv_1 > \uu_1$}.
\end{cases}
\]
Splitting the sum at $\uu_1$ and using the fact that potentials are 
escape probabilities, we have
\begin{align*}
  \sum_{\vv_1=1}^n \ppi^{\Path{n}}_{\vv_1}(\uu_1)
  \le \uu_1 + \sum_{\vv_1=\uu_1+1}^{n} \frac{\uu_1}{\vv_1}
  = O(\uu_1 \log n).
\end{align*}
We similarly obtain an upper bound of $O(\uu_2 \log n)$ in the other dimension.
These bounds along with the initial $O(\log n)$ overhead from
swapping $\uu$ and $\vv$ gives the desired upper bound.
\end{proof}

\subsection{Lower Bounding the Minimum Potential}\label{subsec:lower}

The more involved part of this paper proves a lower bound for the minimum
vertex potential
$
  \min_{\vv \in V\setminus\{ \sink\}} \ppi^{}_{\uu}(\vv)
$
as a function of a fixed vertex $\uu = (\uu_1, \uu_2)$.
Recall that we assumed without loss of generality that $\uu$ is in the top-left
quadrant of $\square{n}$.
We first prove that the minimum potential occurs at vertex $(n,n)$,
the corner farthest from $\uu$.
Using \Cref{lem:swapSourceTarget} to swap $\uu$ and~$(n,n)$ at the expense of a
$\Omega(1/\log{n})$ factor, we reduce the problem to giving a lower bound for
$\ppi^{}_{(n,n)}(\uu)$.
Then we decompose walks
$\WW \in \mathcal{W}^{}(\uu \rightarrow \{(n,n), \sink\})$ into their
one-dimensional walks ${\WW_{1} \in \mathcal{W}^{\Path{n}} (\uu_{1})}$
and $\WW_{2} \in \mathcal{W}^{\Path{n}} (\uu_{2})$, and
we interpret $\ppi^{}_{(n,n)}(\uu)$ as the probability that 
the individual processes $\WW_1$ and $\WW_2$ are present on $n$ at the same
time before either walk leaves the interval $[1,n]$.
Walks on $\line$ that meet at $n$ before leaving the interval $[1,n]$
are equivalent to walks on $\Path{n}$ that meet at $n$ before terminating at
$\sink$.
Lastly, we use conditional probabilities to analyze walks on $\line$
instead of walks on $\Path{n}$ in order to
leverage well-known facts about simple symmetric random walks.

To lower bound the desired probability
$\ppi^{}_{(n,n)}(\uu)$,
we show that a subset of $\mathcal{W}^{}(\uu \rightarrow (n,n))$
of interleaved
one-dimensional walks starting at $\uu_1$ and $\uu_2$ that first
reach $n$ in approximately the same number of steps has a sufficient amount
of probability mass.
We prove this by observing that the distributions of the number of steps
for the walks to first reach $n$ without leaving the interval~$[1,n]$ are
concentrated around $(n-\uu_1)^2$ and $(n-\uu_2)^2$, respectively.
Consequently, we show that this distribution is approximately uniform in an
$\Theta(n^2)$ length interval, with each $t$-step having probability
$\Omega(\uu_1/n^3)$ and $\Omega(\uu_2/n^3)$.
We then use Chernoff bounds to show
that both walks take approximately the same number of steps with
constant probability. Combining these facts, we give the desired lower bound
$\Omega(\uu_1  \uu_2 / n^4)$.

\subsubsection{Opposite Corner Minimizes Potential}

We first show that the corner vertex $(n,n)$ has the minimum
potential up to a constant factor.
Viewing potentials as escape probabilities,
we utilize the geometry of the grid to construct maps between sets of random
walks that prove the potential of an interior vertex is greater than its
axis-aligned projection to the boundary of the grid.
We defer the proof of \Cref{lem:nnIsMin} to \Cref{sec:app_grid}.

\begin{restatable}[]{lemma}{nnIsMin}
\label{lem:nnIsMin}
If $\uu$ is a vertex in the top-left quadrant of $\square{n}$,
then for any non-sink vertex $\vv$ 
we have
  \[
  \ppi^{}_{\uu}\left(\vv\right) \geq
    \frac{1}{16}\ppi^{}_{\uu}\left((n,n)\right).
  \]
\end{restatable}

\subsubsection{Lower Bounding Corner Potential}

By decomposing two-dimensional walks on $\square{n}$ that start at $\uu$
into one-dimensional walks
on $\line$, our lower bound relies on showing that there is a $\Theta(n^2)$
length interval such that each one-dimensional walk of a fixed length in this
interval
has probability $\Omega(\uu_1 / n^3)$ or $\Omega(\uu_2 / n^3)$,
respectively, of remaining above $0$ and reaching $n$ for the first time 
upon termination.
For our purposes, lower bounds for this probability will suffice,
and they follow from the following key property for one-dimensional walks
that we prove in \Cref{sec:path}.


\begin{restatable}[]{lemma}{probLeftRight}
\label{lem:probLeftRight}
Let $n \in \Z_{\ge 1}$ and $1 \leq i \leq \lceil n/2 \rceil$
be any starting position.
For any constant $c > 4$ and any $t \in \Z$ such that
$n^2 / c \le t \le n^2/4$ with $t \equiv n-i \pmod{2}$,
a simple symmetric random walk $\WW$ on $\Z$ satisfies
  \begin{align*}
  \prob{\WW \sim \mathcal{W}^{\line}(i)}{
  \ww^{\left(t\right)} = n,
  \maP{t} = n,
  \text{and } \miP{t} \geq 1 } \geq e^{- 2 - 2c}  \frac{i}{n^3}.
\end{align*}
\end{restatable}

Using \Cref{lem:probLeftRight} with the following lemma,
we give a lower bound for $\ppi^{}_{(n,n)}(\uu)$,
the probability that a walk starting from $\uu$ reaches $(n, n)$ before $\sink$.
\Cref{lem:chernoff} is a consequence of a Chernoff bound, and we
defer its proof to \Cref{sec:app_grid}.

\begin{restatable}[]{lemma}{chernoff}
\label{lem:chernoff}
For all $n \ge 10$, we have
  \[
  \min\left\{
     \frac{1}{2^n} \sum_{\substack{k = \left\lceil \frac{n}{4} \right\rceil \\ k \text{ \normalfont{odd}}}}^{\left\lfloor \frac{3n}{4} \right\rfloor} \binom{n}{k},~
     \frac{1}{2^n} \sum_{\substack{k = \left\lceil \frac{n}{4} \right\rceil \\ k \text{ \normalfont{even}}}}^{\left\lfloor \frac{3n}{4} \right\rfloor} \binom{n}{k}
   \right\} \ge \frac{2}{5}.
 \]
\end{restatable}

\begin{lemma}
\label{lem:voltageLower}
For all $n \geq 10$ and any vertex $\uu$ in the top-left quadrant of
$\square{n}$, we have
\[
  \ppi^{}_{(n,n)}\left( \uu \right)
  \geq e^{-100} \frac{\uu_1  \uu_2}{n^4}.
\]
\end{lemma}

\begin{proof}
We decouple each walk $\WW \in \mathcal{W}^{} ( \uu \rightarrow (n, n))$
into its horizontal walk $\WW_1 \in \mathcal{W}^{\line}(\uu_1)$
and vertical walk $\WW_2 \in \mathcal{W}^{\line}(\uu_2)$.
  The potential $\ppi^{}_{(n,n)}\left( \uu \right)$ can be
interpreted as the probability that $\WW_1$ and $\WW_2$ visit $n$ at the same
time before either leaves the interval $[1,n]$.
We can further decompose $t$-step walks on $\square{n}$ into those that
take $t_1$ steps in the horizontal direction and $t_2$ in the vertical direction.
Considering restricted instances where $\WW_1$ and $\WW_2$ visit $n$ exactly
once, we obtain the following bound by \Cref{lem:normalizingConstant}:
\begin{align}
  \ppi^{}_{(n,n)}\left( \uu \right)
  \ge \sum_{
	   \substack{w \sim \mathcal{W}^{}(\uu \rightarrow (n,n)) \\
	   \text{$w_1$ hits $n$ exactly once} \\
	   \text{$w_2$ hits $n$ exactly once} }} 4^{-|w|}. \label{eqn:prob}
\end{align}
Accounting for all the ways that two one-dimensional walks can be
interleaved,
the right hand side of \Cref{eqn:prob} is 
\begin{align*}
\sum_{t_1, t_2 \geq 0 } &\frac{{{t_1 + t_2}\choose{t_1}}}{4^{t_1+ t_2}}
  \left(\text{$\#$ of $t_1$-step walks from $\uu_1$ that stay in $[1, n-1]$ and terminate at $n$}\right)\\
  &\hspace{1cm}\cdot\left(\text{$\#$ of $t_2$-step walks from $\uu_2$ that stay in $[1, n-1]$ and terminate at $n$}\right).
\end{align*}
Observing that
\begin{align*}
	&\prob{\WW_1 \sim \mathcal{W}^{}\left( \uu_1 \right)}{\ww_1^{(t_1)} = n, \maP{t_1 - 1} = n-1, 
	\miP{t_1-1} \geq 1 }\\
	&\hspace{1.82cm}=
	\frac{\left(\text{$\#$ of $t_1$-step walks from $\uu_1$ that stay in $[1, n-1]$ and terminate at $n$}\right)}{2^{t_1}},
\end{align*}
it follows from \Cref{eqn:prob} that
\begin{align*}
  \ppi^{}_{(n,n)}\left( \uu \right)
  \geq \hspace{-0.2cm}
\sum_{t_1, t_2 \geq 0 } &\frac{{{t_1 + t_2}\choose{t_1}}}{2^{t_1+ t_2}}
  \prob{\WW_1 \sim \mathcal{W}^{}\left( \uu_1 \right)}{\ww_1^{(t_1)} = n, \maP{t_1 - 1} = n-1, 
	\miP{t_1-1} \geq 1 }\\
  &\hspace{0.9cm}\cdot\prob{\WW_2 \sim \mathcal{W}^{}\left( \uu_2 \right)}{\ww_2^{(t_2)} = n,
	\maP{t_2 - 1} = n-1,
	\miP{t_2-1} \geq 1 }.
\end{align*}
	

\noindent
By our choice of $n$ and $\uu$, the right hand side of inequality above equals
\begin{align*}
  \sum_{t_1, t_2 \geq 5} &\frac{{{t_1 + t_2}\choose{t_1}}}{2^{t_1+ t_2}}
\left( \frac{1}{2}\prob{\WW_1 \sim \mathcal{W}^{}\left( \uu_1 \right)}{\ww_1^{(t_1-1)} = n-1,
	\maP{t_1 - 1} = n -1, 
	\miP{t_1-1} \geq 1 }  \right)\\
  &\hspace{0.96cm}\cdot\left( \frac{1}{2}\prob{\WW_2 \sim \mathcal{W}^{}\left( \uu_2 \right)}{\ww_2^{(t_2-1)} = n-1,
		\maP{t_2 - 1} = n-1,
    \miP{t_2-1} \geq 1 }   \right). \tag{$2$}\label{eqn:prob2}
\end{align*}
	


Letting $t = t_1 + t_2$,
we further refine the set of two-dimensional walks so that
$t \in [\frac{1}{5}n^2, \frac{1}{4}n^2]$ and $t_1, t_2 \in [\frac{1}{4}t, \frac{3}{4}t]$
while capturing a sufficient amount of probability mass for a useful lower bound.
Note that the parities of $t_1$ and $t_2$ satisfy
$t_1 \equiv n - \uu_1 \pmod{2}$
and $t_2 \equiv n - \uu_2 \pmod{2}$ for valid walks.
Let $I$ be an indexing of all such pairs $(t_1, t_2)$.
Working from \Cref{eqn:prob2},  we have
\begin{align*}
  \ppi^{}_{(n,n)}\left( \uu \right) &\ge
  \sum_{(t_1, t_2) \in I} \frac{\binom{t_1+t_2}{t_1}}{2^{t_1 + t_2}}
  \left(\frac{1}{2}e^{-2-2(20)} \frac{\uu_1}{n^3} \right)
  \left(\frac{1}{2}e^{-2-2(20)} \frac{\uu_2}{n^3} \right)\\
  &\ge
  e^{-84} \cdot \frac{\uu_1 \uu_2}{4n^6} \sum_{\substack{t \in \left[\frac{n^2}{5}, \frac{n^2}{4} \right] 
                                 \\ t \equiv \uu_1 + \uu_2 \pmod{2}}}
    \frac{2}{5}\\ 
  &\ge
    e^{-84} \cdot \frac{\uu_1 \uu_2}{4n^6} \cdot \frac{n^2}{50} \cdot
        \frac{2}{5} \\
  &\ge  e^{-100} \cdot \frac{\uu_1 \uu_2}{n^4}.
\end{align*}
For the first inequality, we can apply \Cref{lem:probLeftRight} because
\[
  \frac{1}{20}n^2 \le t_1, t_2 \le \frac{3}{16}n^2.
\]
For the second inequality, we group pairs $(t_1, t_2)$ by their sum
$t = t_1 + t_2$ and apply \Cref{lem:chernoff}.
The number of $t \in [\frac{1}{5}n^2, \frac{1}{4}n^2]$ with either parity restriction is at least
$\lfloor \frac{1}{40}n^2 \rfloor \ge \frac{1}{50}n^2$.
\end{proof}

\subsection{Proof of Theorem~\ref{thm:mainGrid}}\label{subsec:thm}
We now combine the upper bound for the sum of potentials given by 
\Cref{lem:sumBound} and the lower bounds in
\Cref{subsec:lower} to obtain the overall upper bound for the transience
class of the grid.


\begin{proof}
For any $\uu = (\uu_1, \uu_2)$ in the top-left quadrant of $\square{n}$, we have
\begin{align*}
  \max_{\uu, \vv \in V\setminus\{\sink\}}
     \left( \sum_{\xx \in V}\pi^{}_{\uu}(\xx) \right)
     \pi^{}_{\uu}(\vv)^{-1}
     &\le \max_{\uu \in V\setminus\{\sink\}} 
  \left(\sum_{\xx \in V} \ppi^{}_{\uu}\left( \xx \right)\right)
  \frac{16}{\ppi^{}_{\uu}\left((n,n) \right)}\\
 &= \max_{\uu \in V\setminus\{\sink\}} 
    \left(\sum_{\xx \in V} \ppi^{}_{\uu}\left( \xx \right)\right)
  \frac{O\left(\log{n}\right)}{\ppi^{}_{(n,n)} \left( \uu\right) }\\
&= 	\max_{\uu \in V\setminus\{\sink\}}
  O\left( \uu_1 \uu_2  \log^3n\right)
  O\left( \frac{n^4\log{n}}{\uu_1 \uu_2} \right)\\
&= O \left(n^4\log^4 n\right).
\end{align*}
The first inequality follows from \Cref{lem:nnIsMin},
the second from \Cref{lem:swapSourceTarget},
and the third from Lemma~\ref{lem:voltageLower} and Lemma~\ref{lem:sumBound}.
The result follows from Theorem~\ref{thm:reduction}.
\end{proof}

\section{Lower Bounding the Transience Class}
\label{sec:lower}

In this section we lower bound $\tcl(\square{n})$ using techniques
similar to those in Section~\ref{sec:grid}. Since the lower bound in
Theorem~\ref{thm:reduction} considers the maximum inverse vertex potential over all pairs
of non-sink vertices $u$ and~$v$, it suffices to upper bound
$\pi^{}_{(n,n)}((1,1))$.
We lower bound vertex potentials by decomposing two-dimensional
walks on $\square{n}$ into one-dimensional walks on $\line$ and then upper
bound the probability that a $t$-step walk on $\line$ starting at~$1$ and
ending at~$n$ does not leave the interval $[1,n]$. More specifically, our upper
bound for 
$\pi^{}_{(n,n)}((1,1))$ follows from \Cref{lem:upperProbLeftRight}
(which we prove in Section~\ref{sec:path}) and
Fact~\ref{fact:binSumInf}.

\begin{restatable}[]{lemma}{upperProbLeftRight}
\label{lem:upperProbLeftRight}
	
	
For all $n \geq 20$ and $t \ge n-1$, we have
\begin{align*}
  &\prob{\WW \sim \mathcal{W}^{\line}(1)}{
      \ww^{(t)} = n, \maP{t} = n, \text{and } \miP{t} \geq 1  }
    \leq \min{\left\{ \frac{e^{25}}{n^3}, 64\left(\frac{n}{t}\right)^3 \right\}}.
\end{align*}
\end{restatable}

\begin{fact}
\label{fact:binSumInf}
For any nonnegative integer $t_1$, we have
\[
  \sum_{t_2 \geq 0} \binom{t_1 + t_2}{t_2}\frac{1}{2^{t_1 + t_2}} = 2.
\]
\end{fact}

\begin{proof}
This follows directly from the negative binomial distribution. Observe that
\begin{align*}
  &\sum_{t_2 \ge 0} \binom{t_1 + t_2}{t_2} \frac{1}{2^{t_1 + t_2}}
  =
  2\sum_{t_2 \ge 0} \binom{(t_1 + 1) - 1 + t_2}{t_2} \frac{1}{2^{t_1 + 1}}\cdot \frac{1}{2^{t_2}}
  = 2,
\end{align*}
as desired.
\end{proof}

By decoupling the two-dimensional walks in a way similar to the proof of
Lemma~\ref{lem:voltageLower}, we apply \Cref{lem:upperProbLeftRight} to the
resulting one-dimensional walks to achieve the desired upper bound.

\begin{lemma}\label{lem:upperBoundMaxSum}
For all $n \ge 20$, we have
  \begin{align*}
    \pi^{}_{(n,n)}((1,1)) &\leq
    2\max_t{\left\{\prob{\WW \sim \mathcal{W}^{}(1)}{
      \ww^{(t)} = n, \maP{t} = n, \miP{t} \geq 1  } \right\}}\\
    & \hspace{0.95cm} \cdot \sum_{t \geq 0} \prob{\WW \sim \mathcal{W}^{}(1)}{
 	\ww^{(t)} = n, \maP{t} = n, \miP{t} \geq 1  }.
	\end{align*}
\end{lemma}

\begin{proof}
  Analogous to our lower bound for $\pi^{}_{(n,n)}((1,1))$,
  decouple each walk ${\WW \in \mathcal{W}^{} ( (1, 1) \rightarrow (n, n))}$
	into its horizontal walk $\WW_1 \in \mathcal{W}^{\line}(1)$
	and its vertical walk $\WW_2 \in \mathcal{W}^{\line}(1)$.
  We view $\ppi^{}_{(n,n)}\left( (1, 1) \right)$ as the probability that
	the walks $\WW_1$ and $\WW_2$ are present on $n$ at the same time
  before either leaves the interval $[1,n]$.
  Letting $t_1$ be the length of $w_1$ and $t_2$ be the length of $w_2$,
  we relax the conditions on the one-dimensional walks and only require that
  $\WW_1$ and $\WW_2$ both are present on $n$ at the final step $t=t_1+t_2$.
  Note that now both walks could have previously been present on $n$ at the same
  time before terminating.  This gives the upper bound
	\begin{align*}
    \ppi^{}_{(n,n)}\left( (1, 1) \right)
	\leq \hspace{-0.2cm}
	\sum_{t_1, t_2 \geq 0 } &\frac{{{t_1 + t_2}\choose{t_1}}}{2^{t_1+ t_2}}
	\prob{\WW_1 \sim \mathcal{W}^{}\left( 1 \right)}{\ww_1^{(t_1)} = n, \maP{t_1} = n, 
		\miP{t_1} \geq 1 }\\
	&\hspace{0.9cm}\cdot\prob{\WW_2 \sim \mathcal{W}^{}\left( 1 \right)}{\ww_2^{(t_2)} = n,
		\maP{t_2 } = n,
		\miP{t_2} \geq 1 }.
	\end{align*}
	
\noindent
Nesting the summations gives
\begin{align*}
  \ppi^{}_{(n,n)}\left( (1, 1) \right) &\leq \sum_{t_1 \geq 0} \prob{\WW_1 \sim \mathcal{W}^{}\left( 1 \right)}{\ww_1^{(t_1)} = n, \maP{t_1} = n, 
		\miP{t_1} \geq 1 } \\
    &\hspace{1.10cm}\cdot \sum_{t_2 \geq 0}\frac{{{t_1 + t_2}\choose{t_1}}}{2^{t_1+ t_2}}\prob{\WW_2 \sim \mathcal{W}^{}\left( 1 \right)}{\ww_2^{(t_2)} = n,
		\maP{t_2 } = n,
		\miP{t_2} \geq 1 }.
\end{align*}	
Using Fact~\ref{fact:binSumInf}, we can upper bound the inner sum by
\begin{multline*}
  \sum_{t_2 \geq 0}\frac{{{t_1 + t_2}\choose{t_1}}}{2^{t_1+ t_2}}\prob{}{\ww_2^{(t_2)} = n,
  \maP{t_2 } = n,
  \miP{t_2} \geq 1 }\\
  \leq 2\max_{t_2}\left\{\prob{}{\ww_2^{(t_2)} = n,
		\maP{t_2 } = n,
    \miP{t_2} \geq 1 }\right\}.
\end{multline*}	
Factoring out this term from the initial expression completes the proof.
\end{proof}

The upper bound on the maximum term in the right hand side of
Lemma~\ref{lem:upperBoundMaxSum} follows immediately from
Lemma~\ref{lem:upperProbLeftRight}. Now we upper bound the summation
in the right hand side of Lemma~\ref{lem:upperBoundMaxSum}
using a simple Lemma~\ref{lem:upperProbLeftRight}.

\begin{lemma}\label{lem:upperBoundSum}
If $n \ge 20$ and
$\WW \sim \mathcal{W}^{\line}(1)$, we have
\begin{align*}
  \sum_{t \geq 0} \prob{}{
		\ww^{(t)} = n, \maP{t} = n, \miP{t} \geq 1  }
  \leq \frac{e^{26}}{n}.
\end{align*}
\end{lemma}

\begin{proof}
We first split the sum into 
\begin{align*}
  &\sum_{t \geq 0} \prob{}{
		\ww^{(t)} = n, \maP{t} = n,  \miP{t} \geq 1  }\\
  &\hspace{0.20cm}= \sum_{n^2 \geq t \geq 0} \prob{}{
		\ww^{(t)} = n, \maP{t} = n,  \miP{t} \geq 1  }
  + \sum_{t > n^2} \prob{}{
	\ww^{(t)} = n, \maP{t} = n, \miP{t} \geq 1  }.
\end{align*}
We will bound both terms by $O(1/n)$.
The upper bound for the first term follows immediately from
Lemma~\ref{lem:upperProbLeftRight} and the fact that we are summing $n^2+1$ terms:
\[
    \sum_{n^2 \geq t \geq 0} \prob{}{
		\ww^{(t)} = n, \maP{t} = n, \miP{t} \geq 1  } \leq \frac{e^{25}}{n}.
\]
To upper bound the second summation, we again use Lemma~\ref{lem:upperProbLeftRight}.
When $t > n^2$, we have
\[
  \prob{}{
		\ww^{(t)} = n, \maP{t} = n, \miP{t} \geq 1  }  \leq 64\left(\frac{n}{t}\right)^3.
\]
Since $64(n/t)^3$ is a decreasing function in $t$, 
\[
  64\left(\frac{n}{t}\right)^3 \leq \int_{t-1}^t 64\left(\frac{n}{t}\right)^3 \mathrm{d}t.
\]
Therefore, we can bound the infinite sum by the integral
\begin{align*}
  \sum_{t > n^2} \prob{}{
    \ww^{(t)} = n, \maP{t} = n, \miP{t} \geq 1  }
  \leq \displaystyle\int_{n^2}^{\infty} 64\left(\frac{n}{t}\right)^3 \mathrm{d}t = \frac{32}{n},
\end{align*}
which concludes the proof since we bounded both halves of the sum by $O(1/n)$.
\end{proof}

\subsection{Proof of Theorem~\ref{thm:mainLowerGrid}}

We can now easily combine the lemmas in this section with the bounds that relate 
vertex potentials to the lower bound for the transience class of $\square{n}$.


\begin{proof}
Applying \Cref{lem:upperBoundMaxSum} and then Lemma~\ref{lem:upperProbLeftRight} and
  Lemma~\ref{lem:upperBoundSum},
it follows that
\begin{align*}
  \pi^{}_{(n,n)}((1,1)) &\leq
  \max_t{\left\{\prob{\WW \sim \mathcal{W}^{}(1)}{
    \ww^{(t)} = n, \maP{t} = n, \miP{t} \geq 1  } \right\}}\\
  & \hspace{0.43cm} \cdot 2\sum_{t \geq 0} \prob{\WW \sim \mathcal{W}^{}(1)}{
\ww^{(t)} = n, \maP{t} = n, \miP{t} \geq 1  }\\
  & \le 2 \cdot \frac{e^{25}}{n^3}\cdot\frac{e^{26}}{n}\\
  &\le \frac{e^{100}}{n^4}.
\end{align*}
\noindent
Therefore, $\pi^{}_{(n,n)}((1,1))^{-1} =\Omega(n^4)$.
By Theorem~\ref{thm:reduction} it follows that $\tcl(\square{n}) = \Omega(n^4)$.
\end{proof}

\section{Simple Symmetric Random Walks}
\label{sec:path}

Our proofs for upper and lower bounding the sandpile transience class on the
grid heavily relied on decoupling two-dimensional walks into two independent
one-dimensional walks since they are easier to analyze.
This claim is immediately apparent when working with vertex potentials
for one-dimensional walks on the path, which we used in the
proof of \Cref{lem:sumBound}.

However, we assumed two essential lemmas about one-dimensional walks to
prove the lower and upper bound of the minimum vertex potential.
Consequently, in this section we examine the probability 
\begin{align}
\prob{\WW \sim \mathcal{W}^{\line}(i)}{
  \ww^{\left(t\right)} \hspace{-0.05cm}= \hspace{-0.03cm}n,
  \maP{t} \hspace{-0.05cm}= \hspace{-0.03cm}n, 
   \miP{t} \hspace{-0.05cm}\geq 1 }\hspace{-0.10cm},\label{eqn:1dwalk}
\end{align}
and we prove these necessary lower and upper bounds in
Lemma~\ref{lem:probLeftRight} and Lemma~\ref{lem:upperProbLeftRight} by
extending previously known properties of simple symmetric random walks on~$\mathbb{Z}$.
The key ideas in these proofs are that:
the position of a walk in one dimension follows the binomial distribution;
the number of walks reaching a maximum position in a fixed number of steps
has an explicit formula;
and there are tight bounds for binomial coefficients via Stirling's approximation.

The properties we need do not immediately follow from previously known facts
because we require conditions on both the minimum and maximum positions.
Section~\ref{sec:maxmin} 
gives proofs of the known explicit expressions for the maximum and minimum position of
a walk, along with several other useful facts that follow from this proof.
In Section~\ref{subsec:lowerBinCoeff} we apply Stirling's bound to give accurate lower
bounds on a range of binomial coefficients. Sections~\ref{subsec:lowerMinPosition}
and~\ref{subsec:lowerFinalAndMax} prove several necessary preliminary lower bound
lemmas. We prove Lemma~\ref{lem:probLeftRight} at the end of
Section~\ref{subsec:lowerFinalAndMax}. In Section~\ref{subsec:upperFinalMaxMin}
we give necessary upper bound lemmas and a proof of
Lemma~\ref{lem:upperProbLeftRight}.

\subsection{Lower Bounding \Cref{eqn:1dwalk}}

To lower bound \Cref{eqn:1dwalk}, we split the desired probability
into the product of two probabilities using the definition of conditional
probability. Then we prove lower bounds for each. 
\begin{itemize}
\item In Lemma~\ref{lem:probLeft} we show for $t \in \Theta(n^2)$,
  the probability that a walk on $\Z$ starting at ${1 \leq i \leq \lceil n/2 \rceil}$
is
$
\prob{\WW \sim \mathcal{W}^{}(i)}{\miP{t} \geq 1} = \Omega \left( i/n \right).
$
\item In Lemma~\ref{lem:leftWorse} and Lemma~\ref{lem:probRight} we bound the probability
that a walk starting at $1 \leq i \leq \lceil n /2 \rceil$
of length $t \in \Theta(n^2)$ reaches $n$ at step $t$ without going above $n$,
conditioned on never dropping below $1$:
\begin{align*}
\prob{}{
\ww^{\left(t\right)} = n,
    \maP{t} =n
  \left|\:\miP{t} \geq 1 \right.} = \Omega \left( \frac{1}{n^2} \right).
\end{align*}
\end{itemize}

Lemma~\ref{lem:probLeftRight} immediately follows multiplying these
two bounds together.
This division allows us to separate proving a minimum and maximum,
and in turn simplifies applying known bounds on binomial distributions.
Specifically, Lemma~\ref{lem:probLeft} is an immediate consequence
of explicit expressions for the minimum point of a walk and bounds
on binomial coefficients, both of which will be given rigorous
treatment in Section~\ref{sec:maxmin}.

These proofs will also output a known explicit expression for
the probability of the walk reaching~$n$ at step $t$,
while only staying to its left.
All that remains then is to condition the walk
to not go to the left of~$1$.
Note that $1$ is in the opposite direction of $n$, with respect to the starting position~$i$.
We formally show that the probability of reaching $n$ without going above $n$ only
improves if the walk cannot move too far
in the wrong direction, but only for $t \leq (n-i+1)^2$, thus giving the reason we need to upper bound $t$ by $n^2/4$.

\subsection{Upper Bounding \Cref{eqn:1dwalk}}

The desired lemma only concerns walks
starting at $i = 1$, which will be critical for our proof. The key idea will
then be to split the walk in half and consider the probability that the
necessary conditions are satisfied in the first $t/2$ steps and in the second
$t/2$ steps. The midpoint of the walk at $t/2$ steps can be any point in the
interval $[1,n]$, so we must sum over all these possible midpoints. Removing
the upper and lower bound conditions, respectively, will then give the upper
bound in Lemma~\ref{lem:divideWalkinHalf}:

\begin{align*}
  &\prob{\WW \sim \mathcal{W}^{}(1)}{
    \ww^{(t)} = n, \maP{t} = n, \miP{t} \geq 1  }\\
  &\hspace{2.0cm}\leq \sum_{i=1}^n \prob{\WW \sim \mathcal{W}^{}(1)}{
  \ww^{(\lfloor \frac{t}{2} \rfloor)} = i, \miP{\lfloor \frac{t}{2} \rfloor} \geq 1  }
  \cdot\prob{\WW \sim \mathcal{W}^{}(i)}{
  \ww^{(\lceil \frac{t}{2} \rceil)} = n, \maP{\lceil \frac{t}{2} \rceil} = n }.
\end{align*}

Due to the first $t/2$ step walk starting at $1$ and the second $t/2$ step walk
ending at $n$, the conditions $\miP{t} \geq 1$ for the first walk and
$\maP{\lceil \frac{t}{2} \rceil} = n$ for the second walk will be the 
difficult property for each walk to satisfy, respectively. Next we apply 
facts proved in Section~\ref{sec:maxmin} to obtain expressions for each term
within the summation. The remainder of the upper bound analysis will then
focus on bounding those expressions.


\subsection{Maximum Position of a Walk}\label{sec:maxmin}

As previously mentioned, our proofs mostly leverage well-known facts
about the maximum/minimum position of a random walk along with corresponding
bounds for these probabilities.
This section will first give the result regarding
maximum/minimum position of walks and a connection to Stirling's approximation.


Observe that if we are only concerned with a single end point,
we can fix the starting location at $0$ by shifting accordingly.
In these cases, the following bounds are well known in combinatorics.

\begin{fact}[{\cite{Rota79}}]
\label{fact:numWalks}
For any $t, n \in \Z_{\ge 0}$, we have
\begin{align*}
  \prob{\WW \sim \mathcal{W}^{}(0)}{\maP{t} = n}
  =
  \begin{cases}
    \prob{}
      {\ww^{(t)} = n} = {t \choose \frac{t + n}{2}}\frac{1}{2^t}
      & \text{if $t + n \equiv 0 \pmod*{2}$},\\
    \prob{}
      {\ww^{(t)} = n + 1} = {t \choose \frac{t + n + 1}{2}}\frac{1}{2^t}
      & \text{if $t + n \equiv 1 \pmod*{2}$}.
  \end{cases}
\end{align*}
\end{fact}

\begin{proof}
For any $k \leq n$, consider a walk
$\WW \in \mathcal{W}^{}\left(0\right)$ that satisfies
$w^{(t)} = k$ and $\maP{t} \geq n$.
Let $t^*$ be the first time that $w^{(t^*)} = n$, and
construct the walk $m$ ending at $2n - k$ such that
\[
  m^{\left(\hat{t}\right)} =
  \begin{cases}
    w^{\left(\hat{t}\right)} & \text{if $0 \le \hat{t} \le t^*$},\\
    2n - w^{\left(\hat{t}\right)} & \text{if $t^* < \hat{t} \le t$}.
  \end{cases}
\]
This reflection map is a bijection, so for $k \le n$ we have
\begin{align*}
  \prob{\WW \sim \mathcal{W}^{}(0)}
  {\ww^{\left(t\right)} = k, \maP{t} \geq n}
  =
  \prob{\WW \sim \mathcal{W}^{}(0)}
  {\ww^{\left(t\right)} = 2n - k}.
\end{align*}
Subtracting the probability of the maximum position being at least $n+1$ gives
\begin{multline*}
  \prob{\WW \sim \mathcal{W}^{\line}(0)}
  {\ww^{\left(t\right)} = k \text{ and } \maP{t} = n}\\
  =
  \prob{\WW \sim \mathcal{W}^{\line}(0)}
  {\ww^{\left(t\right)} = 2n - k}
  -
  \prob{\WW \sim \mathcal{W}^{\line}(0)}
  {\ww^{\left(t\right)} = 2(n+1) - k}.
\end{multline*}
Summing over all $k \le n$, we have
\[
 \prob{\WW \sim \mathcal{W}^{\line}(0)}{\maP{t} = n}
 = 
  \prob{\WW \sim \mathcal{W}^{\line}(0)}
  {\ww^{(t)} = n}
  + \prob{\WW \sim \mathcal{W}^{\line}(0)}
  {\ww^{(t)} = n + 1}.
\]
Considering the parity of $t$ and $n$ completes the proof.
%
%
%
\end{proof}


The proof above contains two intermediate expressions for probabilities similar
to the ones we want to bound.

\begin{fact}\label{fact:goesAbove}
For any integers $n \geq 0$ and $k \leq n$, we have
\begin{align*}
  \prob{\WW \sim \mathcal{W}^{}(0)}
  {\ww^{\left(t\right)} = k, \maP{t} \geq n}
  =
  \prob{\WW \sim \mathcal{W}^{}(0)}
  {\ww^{\left(t\right)} = 2n - k}.
\end{align*}
\end{fact}

\begin{fact}\label{fact:endn}
Let $t, n \in \Z_{\ge 0}$. For any integer $k \le n$, 
\begin{align*}
  \prob{\WW \sim \mathcal{W}^{}(0)}
  {\ww^{\left(t\right)} = k, \maP{t} = n}
  =
    \begin{cases}
      \binom{t}{\frac{t + 2n - k}{2}}\frac{1}{2^t}\left(\frac{4n - 2k + 2}{t + 2n - k + 2}\right)
    & \text{if $t + k \equiv 0 \pmod*{2}$},\\
    0
    & \text{if $t + k \equiv 1 \pmod*{2}$}.
    \end{cases}
\end{align*}
\end{fact}

\begin{proof}
Using Fact~\ref{fact:numWalks} and analyzing the parity of the walks gives
\begin{align*}
  {t \choose \frac{t + 2n - k}{2}} \frac{1}{2^t} - {t \choose \frac{t + 2n - k+2}{2}} \frac{1}{2^t}
  &= {t \choose \frac{t + 2n - k}{2}} \frac{1}{2^t} - \frac{t - 2n +k}{t + 2n - k+2}{t \choose \frac{t + 2n - k}{2}} \frac{1}{2^t}\\
  &= {t \choose \frac{t + 2n - k}{2}} \frac{1}{2^t} \left(\frac{4n - 2k + 2}{t + 2n - k + 2} \right),
\end{align*}
as desired.
\end{proof}

\subsection{Lower Bounding Binomial Coefficients}\label{subsec:lowerBinCoeff}

Ultimately, our goal is to give strong lower bounds on closely
related probabilities to the ones above.
To do so, we need to use various bounds on binomial coefficients
that are consequences of Stirling's approximation.

\begin{fact}[Stirling's Approximation]
\label{fact:stirling}
For any positive integer $n$, we have
\[
  \sqrt{2\pi}
  \leq
  \frac{n!}{\sqrt{n}\left(\frac{n}{e}\right)^n}
  \leq
  e.
\]
\end{fact}

An immediate consequence of this is a concentration bound
on binomial coefficients.

\begin{fact}
\label{fact:binBound}
Let $c, n \in \R_{> 0}$ such that $c \sqrt{n} < n$.
For any $k \in [(n - c\sqrt{n})/2, (n + c\sqrt{n})/2]$, we have
\[
  \binom{n}{k}
  \geq
  e^{-1-c^2} \cdot \frac{2^n}{\sqrt{n}}.
\]	
\end{fact}

\begin{proof}
We directly substitute Stirling's approximation to the definition of binomial
coefficients to get
\begin{align*}
\binom{n}{\frac{n-c\sqrt{n}}{2}}
  &=\frac{n!}{\left(\frac{n-c\sqrt{n}}{2}\right)!
			\left(\frac{n+c\sqrt{n}}{2}\right)!}\\
      &\ge\frac{\sqrt{2\pi n}\left(\frac{n}{e}\right)^n}
	{e\sqrt{\frac{n-c\sqrt{n}}{2}}
	\left(\frac{n-c\sqrt{n}}{2e}\right)^\frac{n-c\sqrt{n}}{2}
  e\sqrt{\frac{n+c\sqrt{n}}{2}}
	\left(\frac{n+c\sqrt{n}}{2e}\right)^\frac{n+c\sqrt{n}}{2}}\\
  &\ge\frac{2\sqrt{2\pi}}{e^2 \sqrt{n}} 
   \cdot \frac{2^n}{\left(1-\frac{c^2}{n}\right)^\frac{n}{2}
\left(1-\frac{c}{\sqrt{n}}\right)^{-\frac{c\sqrt{n}}{2}}
\left(1+\frac{c}{\sqrt{n}}\right)^\frac{c\sqrt{n}}{2}}\\
  &\ge\frac{2\sqrt{2\pi}}{e^{2 + c^2}} \cdot \frac{2^n}{\sqrt{n}}\\
  &\ge e^{-1-c^2} \cdot \frac{2^n}{\sqrt{n}},
	\end{align*}
as desired.
\end{proof}

\subsection{Lower Bounding the Minimum Position}\label{subsec:lowerMinPosition}

We now bound the probability of the minimum position of a walk in
$\mathcal{W}^{}(i)$ being at least $1$ after $t$ steps.

\begin{lemma}
\label{lem:probLeft}
	
For any positive integer $n$, 
initial position $1 \le i \le \lceil n/2 \rceil$,
and constant $c > 4$,
if we have $t \in [n^2/c, n^2/4]$, then
\[
\prob{\WW \sim \mathcal{W}^{}(i)}{\miP{t} \geq 1}
\geq e^{-1 - c} \cdot \frac{i}{n} .
\]
\end{lemma}

\begin{proof} 
First observe that
\[
\prob{\WW \sim \mathcal{W}^{}(i)}{\miP{t} \geq 1}
  \hspace{-0.05cm}=\hspace{-0.05cm}\sum_{k=1}^{i} \prob{\WW \sim \mathcal{W}^{}(i)}{\miP{t} = k}.
\]
By symmetry, this sum is
\[
\sum_{k=0}^{i-1} \prob{\WW \sim \mathcal{W}^{}(0)}{\maP{t} = k}.
\]
For each $0 \le k \le i-1$, Fact~\ref{fact:numWalks} implies that
\[
\prob{\WW \sim \mathcal{W}^{}(0)}{\maP{t} = k}
  \in \left\{\binom{t}{\frac{t + k}{2}} \frac{1}{2^t},
  \binom{t}{\frac{t + k + 1}{2}} \frac{1}{2^t}\right\}.
\]
By assumption
$k \le k + 1 \le i \le n \le \sqrt{ct}$,
so applying Fact~\ref{fact:binBound} gives
\begin{align*}
  \min\left\{
    \binom{t}{\frac{t+k}{2}} \frac{1}{2^t},~
    \binom{t}{\frac{t+k+1}{2}} \frac{1}{2^t}
  \right\}
  &\ge
  \binom{t}{\frac{t + \sqrt{ct}}{2}} \frac{1}{2^t}\\
  &\ge e^{-1 - c} \frac{1}{\sqrt{t}}\\
  &\ge e^{-1 - c} \frac{1}{n},
\end{align*}
because $t \le n^2/4$.
Summing over $0 \le k \le i-1$ gives the desired bound.
\end{proof}

\subsection{Lower Bounding the Final and Maximum Position}\label{subsec:lowerFinalAndMax}

Similarly, we can use binomial coefficient approximations to bound the
probability of a $t$-step walk terminating at $n$ while never moving to a
position greater than $n$.

\begin{lemma}
\label{lem:probRight}
For any initial position $1 \le i \le \lceil n/2 \rceil$ and
any $\max\{n, n^2/c \}
\le t \le n^2/4$ with $t \equiv n-i \pmod{2}$, 
we have
\[
\prob{\WW \sim \mathcal{W}^{}(i)}
{\maP{t} = n, \ww^{\left( t \right)} = n}
\geq e^{-1-c} \cdot \frac{1}{n^2}.
\]
\end{lemma}


\begin{proof}
By symmetry we rewrite the probability as
\[
\prob{\WW \sim \mathcal{W}^{}(0)}
{\maP{t} = n-i, \ww^{\left( t \right)} = n-i}.
\]
Fact~\ref{fact:endn} gives that this probability equals to
\[
\frac{1}{2^t}\binom{t}{\frac{t + n - i}{2}}
\frac{2(n - i + 1)}{t + n - i + 2}.
\]
We can separately bound the last two terms according to the assumptions
  on $t$ and $i$. Setting $i = 0$ minimizes $\binom{t}{(t + n - i)/2}$
for all $i \geq 0$. Setting $i = \lceil n/2 \rceil$ in the numerator, $i = 0$
  in the denominator, and $t = n^2/4$ minimizes $(2(n-i + 1))/(t + n + 2)$.  It
  follows that
\[
\frac{1}{2^t}\binom{t}{\frac{t+n}{2}} \cdot \frac{2(\lfloor n/2 \rfloor +1)}{n^2/4+n+2}
\ge \frac{1}{2^t}\binom{t}{\frac{t+n}{2}} \frac{n}{n^2}.
\]
We reapply Fact~\ref{fact:binBound} with the
observation that $n  \leq \sqrt{ct}$ to get
\[
\frac{1}{2^t}\binom{t}{\frac{t+\sqrt{ct}}{2}} \frac{1}{n}
\ge e^{-1-c} \cdot \frac{1}{n^2},
\]
as desired.
\end{proof}

It remains to condition upon the minimum of a walk.
This hinges upon the following statement about moving in the
wrong direction only decreasing the probability a walk starting
at some $1 \leq i \leq \lceil n/2 \rceil$ ending at $n$ without
ever going past $n$.

\begin{lemma}
\label{lem:leftWorse}
For any $1 \le i \le \lceil n/2 \rceil$,
at any step $t \leq n^2/4$ with $t \equiv n-i \pmod{2}$, we have
\begin{align*}
  \prob{\WW \sim \mathcal{W}^{}(i)}
{\ww^{\left(t\right)}=n,
\maP{t}=n}
  \geq
\prob{\WW \sim \mathcal{W}^{}(i)}
{\ww^{\left(t\right)}=n,
\maP{t}=n
  \left|\: \miP{t} < 1\right.}.
\end{align*}
\end{lemma}

\begin{proof}
Condition on $\miP{t} < 1$ and consider the first time $\widehat{t}$ the walk hits $0$.
This means $i
\equiv \widehat{t} \pmod 2$ and in turn $n \equiv t - \widehat{t} \pmod 2$.
The probability of $\maP{t}  = \ww^{(t)} = n$
via the walk in steps $\widehat{t} + 1,\dots, t$ is then
at most
\[
\prob{\WW \sim \mathcal{W}^{}(0)}
{\ww^{\left(t - \widehat{t} \right)} =
\maP{t - \widehat{t}} = n}.
\]
Note that we have inequality since it is possible that
we already have $\maP{\widehat{t}}  > n$.
Therefore, it suffices to show for any $n$ and any
$1 \leq \widehat{t} \leq t$ we have
\begin{align*}
  \prob{\WW \sim \mathcal{W}^{}(0)}
{\ww^{\left(t - \widehat{t} \right)} =
	\maP{t - \widehat{t}} = n}
  \leq
\prob{\WW \sim \mathcal{W}^{}(i)}
{\ww^{\left(t\right)} = \maP{t} = n}.
\end{align*}

There are two variables that are shifted from one side of the inequality to the
  other, the starting position of the walk and the number of steps. In order to
  prove the inequality, we will show that both taking more steps and
  starting further to the right will only improve the probability of ending at
  $n$ and not going above $n$.

We begin by showing that taking more steps will only improve this probability:
\begin{align*}
  &\prob{\WW \sim \mathcal{W}^{}(0)}
  {\ww^{\left(t - \widehat{t} \right)} =
    \maP{t - \widehat{t}} = n}\\
  &\hspace{2.5cm}\leq
  \max\Bigg\{\prob{\WW \sim \mathcal{W}^{}(0)}
  {\ww^{\left(t-1\right)} = \maP{t - 1} = n},
  \prob{\WW \sim \mathcal{W}^{}(0)}
  {\ww^{\left(t\right)} = \maP{t} = n} \Bigg\}.
\end{align*}

\noindent
There is no guarantee that $t \equiv n \pmod 2$, so we consider $t$ or
$t-1$ steps depending on parity. We are guaranteed that $t-1 \geq t -
\widehat{t}$ since $\widehat{t} \geq 1$, so without loss of generality, we
assume $t \equiv \widehat{t} \pmod 2$ and show 
\begin{align*}
  \prob{\WW \sim \mathcal{W}^{}(0)}
{\ww^{\left(t - \widehat{t} \right)} =
	\maP{t - \widehat{t}} = n}
  \leq
\prob{\WW \sim \mathcal{W}^{}(0)}
{\ww^{\left(t\right)} = \maP{t} = n}.
\end{align*}
Note that the proof is equivalent when $t - 1 \equiv \widehat{t} \pmod 2$.

Using Fact~\ref{fact:endn}, we have the explicit probability
\[
\prob{\WW \sim \mathcal{W}^{}(0)}
{\maP{t} = \ww^{\left(t\right)}=n}
= \frac{1}{2^t}{t \choose \frac{t + n}{2}}\frac{2n+2}{t+n+2}.
\]
Substituting $t$ by $t-2$ into the equation above and comparing the right hand sides gives
\begin{align*}
  \prob{\WW \sim \mathcal{W}^{}(0)}
{\ww^{\left(t-2\right)} = \maP{t-2} = n}
  \leq
\prob{\WW \sim \mathcal{W}^{}(0)}
{\ww^{\left(t\right)} = \maP{t} = n},
\end{align*}
because we know by assumption that
\begin{align*}
\frac{1}{2^{t-2}}{t-2 \choose \frac{t+n}{2}-1}\frac{2n+2}{t+n} & \le \frac{1}{2^t}{t \choose \frac{t+n}{2}}\frac{2n+2}{t+n+2}\\
\frac{(t-2)!}{\left(\frac{t+n-2}{2}\right)!\left(\frac{t-n-2}{2}\right)!}\frac{1}{t+n}&\le \frac{1}{4}\frac{t!}{\left(\frac{t+n}{2}\right)!\left(\frac{t-n}{2}\right)!}\frac{1}{t+n+2}\\
\frac{1}{t+n} &\le \frac{t(t-1)}{(t+n)(t-n)(t+n+2)}\\
3t&\leq n^2 + 2n.
\end{align*}
Inductively applying this argument inductively for $t - 2$ proves the inequality.

To complete our proof, it now suffices to show
\begin{align*}
&\max\bigg\{\prob{\WW \sim \mathcal{W}^{}(0)}
{\ww^{\left(t-1\right)} = \maP{t - 1} = n},\prob{\WW \sim \mathcal{W}^{}(0)}
{\ww^{\left(t\right)} = \maP{t} = n} \bigg\} \\
  &\hspace{8.5cm}\leq \prob{\WW \sim \mathcal{W}^{}(i)}
{\ww^{\left(t\right)} =
	\maP{t} = n},
\end{align*}
which we prove similarly.
First rewrite the right hand side using the fact that
\begin{align*}
  &\prob{\WW \sim \mathcal{W}^{}(i)}
  {\ww^{\left(t\right)} = \maP{t} = n}
  = \prob{\WW \sim \mathcal{W}^{}(0)} {\ww^{\left(t\right)} = \maP{t} = n-i},
\end{align*}
and initially assume that $t \equiv n \pmod 2$, which implies $n \equiv n - i \pmod 2$. Again, using the explicit formula from Fact~\ref{fact:endn} and substituting $n$ by $n - 2$ gives 
\begin{align*}
  &\prob{\WW \sim \mathcal{W}^{}(0)}
	{\ww^{\left(t\right)} = \maP{t} = n}
  \leq \prob{\WW \sim \mathcal{W}^{}(0)} {\ww^{\left(t\right)} = \maP{t} = n-2},
\end{align*}
when $t + 2 \leq n^2$, 
which true by assumption and can be inductively applied until $n = (n-i+2)$
because $(n - i+2) \geq \lceil n/2 \rceil + 1$. Unfortunately, we cannot
entirely apply the same proof when $t -1 \equiv n \pmod 2$ because this implies
$n \not\equiv n-i \pmod 2$. Applying the same proof as for $t \equiv n \pmod 2$
we can obtain
\begin{align*}
  &\prob{\WW \sim \mathcal{W}^{}(0)} {\ww^{\left(t - 1\right)} = \maP{t-1}
= n}
  \leq \prob{\WW \sim \mathcal{W}^{}(0)} {\ww^{\left(t-1\right)} =
\maP{t-1} = n-i+1},
\end{align*}
because $(t-1) + 2 \leq (n-i+3)^2$.
	
Therefore, we can conclude the proof by showing
\begin{align*}
  &\prob{\WW \sim \mathcal{W}^{}(0)} {\ww^{\left(t-1\right)} = \maP{t-1} = n-i+1}
  \leq \prob{\WW \sim \mathcal{W}^{\line}(0)}
		{\ww^{\left(t\right)} = \maP{t} = n - i}.
\end{align*}
This is then true when
\[
  n-i \leq \frac{t}{t-(n-i)} \cdot (n-i+1),
\]
which holds for $n-i \geq 0$.
\end{proof}

An immediate corollary of this \Cref{lem:leftWorse} is that if we condition
on the walk not going to the left of $1$, it only becomes more
probable to reach~$n$ without going above $n$.
Now we prove the main result of this section.
\probLeftRight*

\begin{proof} 
Consider any starting position $1 \le i \le \lceil n/2 \rceil$
and any time $n^2/c \le t \le n^2/4$ 	with $t \equiv n-i \pmod{2}$.
By the definition of conditional probability we have
\begin{align*}
  \prob{\WW \sim \mathcal{W}^{}(i)}{
  \ww^{(t)} = n, \maP{t} = n, \miP{t} \geq 1  }
 & =
  \prob{\WW \sim \mathcal{W}^{}(i)}
  {\ww^{\left(t\right)}=n,
  \maP{t}=n
  \left|\: \miP{t} \geq 1\right.}\\
  &\hspace{0.25cm}\cdot 
  \prob{\WW \sim \mathcal{W}^{}(i)}
  {\miP{t} \geq 1}.
\end{align*}

Lemma~\ref{lem:probLeft} shows that the second term is at
least $\exp(-1 - c)) i / n$.
Taking the probability under $\miP{t} \geq 1$
(i.e., the complementary event of $\miP{t} < 1$) in 
Lemma~\ref{lem:leftWorse}
allows us to upper bound the first term using Lemma~\ref{lem:probRight} by
\begin{align*}
  \prob{\WW \sim \mathcal{W}^{}(i)}
  {\ww^{\left(t\right)}=n,
   \maP{t}=n
  \left|\: \miP{t} \geq 1\right.}
  &\geq
  \prob{\WW \sim \mathcal{W}^{}(i)}
  {\ww^{\left(t\right)}=n,
   \maP{t}=n}\\
  &\geq e^{-1 - c} \frac{1}{n^2}.
\end{align*}
Putting these together then gives
\begin{align*}
  \prob{\WW \sim \mathcal{W}^{}(i)}{
 	\ww^{(t)} = n, \maP{t} = n, \miP{t} \geq 1  }
&\geq \left( e^{-1 - c} \frac{i}{n} \right)
\cdot \left( e^{-1 - c} \frac{1}{n^2} \right)\\
  &= e^{-2 - 2c}\cdot \frac{i}{n^3},
\end{align*}
which completes the proof.
\end{proof}

\subsection{Upper Bounding the Final, Maximum, and Minimum Position}\label{subsec:upperFinalMaxMin}

We begin by splitting every $t$ step walk in half, and instead consider the
probability of each walk satisfying the given conditions. In order to give
upper bounds of these probabilities, we will relax the requirements, allowing
us to more easily relate the probabilities to previously known facts about
one-dimensional walks that we proved in Section~\ref{sec:maxmin}.
Furthermore, by splitting the walk in half we now have to consider all possible
midpoints in $[1,n]$.
\begin{lemma}\label{lem:divideWalkinHalf}
For all integers $1 \le n \le t$, we have
\begin{align*}
  &\prob{\WW \sim \mathcal{W}^{}(1)}{
  \ww^{(t)} = n, \maP{t} = n, \miP{t} \geq 1  }\\
&\hspace{2.0cm}\leq \sum_{i=1}^n \prob{\WW \sim \mathcal{W}^{}(1)}{
  \ww^{(\lfloor \frac{t}{2} \rfloor)} = i, \miP{\lfloor \frac{t}{2} \rfloor} \geq 1  }
 \cdot\prob{\WW \sim \mathcal{W}^{}(i)}{
  \ww^{(\lceil \frac{t}{2} \rceil)} = n, \maP{\lceil \frac{t}{2} \rceil} = n }.
\end{align*}
\end{lemma}

\begin{proof}
By subdividing the walk roughly in half, we consider all possible positions of
a walk after half of its steps such that the walk satisfies the maximum and minimum
conditions. The second half of the walk must end at $n$, which implies the
maximum position of the walk must be at least~$n$. Thus, the first half of the
walk only needs to not go above $n$. Accordingly,
we can write
\begin{align*}
  &\prob{\WW \sim \mathcal{W}^{}(1)}{
  \ww^{(t)} = n, \maP{t} = n, \miP{t} \geq 1  }\\
&\hspace{3.0cm}=\sum_{i=1}^n \prob{\WW \sim \mathcal{W}^{}(1)}{
    \ww^{(\lfloor \frac{t}{2} \rfloor)} = i, \maP{\lfloor \frac{t}{2} \rfloor} \leq n, \miP{\lfloor \frac{t}{2} \rfloor} \geq 1}\\
  &\hspace{6.0cm}\cdot \prob{\WW \sim \mathcal{W}^{}(i)}{
    \ww^{(\lceil \frac{t}{2} \rceil)} = n, \maP{\lceil \frac{t}{2} \rceil} = n, \miP{\lceil \frac{t}{2} \rceil} \geq 1 }.
\end{align*}
Removing conditions that the walks must satisfy cannot decrease
the probability, so our upper bound follows.
\end{proof}

From Fact~\ref{fact:endn} we can obtain explicit expressions for each inner
term of the summation, which we then simplify into a strong bound on the
summation in the following lemma.

\begin{lemma}\label{lem:sumBoundForMaxMinEnd}
For all integers $1 \le n \le t$, we have
\begin{align*}
  &\prob{\WW \sim \mathcal{W}^{}(1)}{
  \ww^{(t)} = n, \maP{t} = n, \miP{t} \geq 1  }\\
&\hspace{4.5cm}\leq \sum_{i=1}^n \left(\frac{16i(n-i+1)}{t^2}\right) \binom{\lfloor \frac{t}{2} \rfloor}{\frac{\lfloor \frac{t}{2} \rfloor+i-1}{2}} \frac{1}{2^{\lfloor \frac{t}{2} \rfloor}}
  \cdot\binom{\lceil \frac{t}{2} \rceil}{\frac{\lceil \frac{t}{2} \rceil+(n-i+1)-1}{2}} \frac{1}{2^{\lceil \frac{t}{2} \rceil}}.
\end{align*}
\end{lemma}

\begin{proof}
Apply the upper bound from Lemma~\ref{lem:divideWalkinHalf} and
examine each inner term in the summation. By the symmetry of walks, there
must be an equivalent number of $\lfloor t/2 \rfloor$ step walks with
endpoints~$1$ and $i$ that never walk below $1$ versus those that never walk
above $i$. Thus
\begin{align*}
 \prob{\WW \sim \mathcal{W}^{}(1)}{ \miP{\lfloor \frac{t}{2} \rfloor}
   \geq 1,	\ww^{(\lfloor \frac{t}{2} \rfloor)} = i }
  = \prob{\WW
\sim \mathcal{W}^{}(1)}{ \maP{\lfloor \frac{t}{2} \rfloor} \leq i,	\ww^{(\lfloor \frac{t}{2} \rfloor)} = i }.
\end{align*}
Shifting the start of the walk to $0$ allows us to apply Fact~\ref{fact:endn},
because $\maP{\lfloor \frac{t}{2} \rfloor} \leq i$ is equivalent to
$\maP{\lfloor \frac{t}{2} \rfloor} = i$ if the walk must end at $i$.
Therefore,
\begin{align*}
  \prob{\WW \sim \mathcal{W}^{}(1)}{
		\miP{t} \geq 1,	\ww^{(\lfloor \frac{t}{2} \rfloor)} = i }
  = \left(\frac{2i}{\lfloor \frac{t}{2} \rfloor+i+1}\right) \binom{\lfloor \frac{t}{2} \rfloor}{\frac{\lfloor \frac{t}{2} \rfloor+i-1}{2}} \frac{1}{2^{\lfloor \frac{t}{2} \rfloor}},
\end{align*}
when the parity is correct and $0$ otherwise. This works as an upper bound.
Similarly, by shifting the start to $0$ and applying Fact~\ref{fact:endn} we have
\begin{align*}
  \prob{\WW \sim \mathcal{W}^{}(i)}{
		\ww^{(\lceil \frac{t}{2} \rceil)} = n, \maP{\lceil \frac{t}{2} \rceil} = n }
  = 	\left(\frac{2(n-i+1)}{\lceil \frac{t}{2} \rceil+(n-i + 1)+1}\right) \binom{\lceil \frac{t}{2} \rceil}{\frac{\lceil \frac{t}{2} \rceil+(n-i+1)-1}{2}} \frac{1}{2^{\lceil \frac{t}{2} \rceil}}.
\end{align*}
	
Applying Lemma~\ref{lem:divideWalkinHalf}, we now have expressions for the term inside the summation, so
	\begin{align*}
    &\prob{\WW \sim \mathcal{W}^{}(1)}{
		\ww^{(\lfloor \frac{t}{2} \rfloor)} = i, \miP{\lfloor \frac{t}{2} \rfloor} \geq 1  }
  \cdot\prob{\WW \sim \mathcal{W}^{}(i)}{
		\ww^{(\lceil \frac{t}{2} \rceil)} = n, \maP{\lceil \frac{t}{2} \rceil} = n }\\ 
  &= \left(\frac{2i}{\lfloor \frac{t}{2} \rfloor+i+1}\right) \binom{\lfloor \frac{t}{2} \rfloor}{\frac{\lfloor \frac{t}{2} \rfloor+i-1}{2}} \frac{1}{2^{\lfloor \frac{t}{2} \rfloor}}
    \cdot \left(\frac{2(n-i+1)}{\lceil \frac{t}{2} \rceil+(n-i + 1)+1}\right) \binom{\lceil \frac{t}{2} \rceil}{\frac{\lceil \frac{t}{2} \rceil+(n-i+1)-1}{2}} \frac{1}{2^{\lceil \frac{t}{2} \rceil}}\\ 
    &\leq \left(\frac{16i(n-i+1)}{t^2}\right) \binom{\lfloor \frac{t}{2} \rfloor}{\frac{\lfloor \frac{t}{2} \rfloor+i-1}{2}} \frac{1}{2^{\lfloor \frac{t}{2} \rfloor}}
    \cdot\binom{\lceil \frac{t}{2} \rceil}{\frac{\lceil \frac{t}{2} \rceil+(n-i+1)-1}{2}} \frac{1}{2^{\lceil \frac{t}{2} \rceil}},
\end{align*}
which we upper bound by the fact that $(\lfloor \frac{t}{2} \rfloor+i+1)(\lceil \frac{t}{2} \rceil+(n-i + 1)+1) \geq t^2/4$.
This completes the proof.
\end{proof}

The following lemma gives a upper bound for the inner expression from
Lemma~\ref{lem:sumBoundForMaxMinEnd} by bounding the binomial
coefficients with the central binomial coefficients and
using Stirling's approximation.

\begin{lemma}\label{lem:simpleMaxforProductProb}
For any integer $1 \le i n$, we have
\begin{align*}
  \left(\frac{16i(n-i+1)}{t^2}\right)
    \binom{\lfloor \frac{t}{2} \rfloor}{\frac{\lfloor \frac{t}{2} \rfloor+i-1}{2}} \frac{1}{2^{\lfloor \frac{t}{2} \rfloor}}
\binom{\lceil \frac{t}{2} \rceil}{\frac{\lceil \frac{t}{2} \rceil+(n-i+1)-1}{2}} \frac{1}{2^{\lceil \frac{t}{2} \rceil}}
  \leq 64\frac{n^2}{t^3}.
\end{align*}
\end{lemma}

\begin{proof}
Given that $1 \le i \le n$, we can crudely upper bound $i(n-i+1)$ by $n^2$.
Additionally, we can will use Stirling's approximation on the central binomial
coefficient to upper bound our binomial coefficients.
\[
    \binom{\lceil \frac{t}{2} \rceil}{\frac{\lceil \frac{t}{2} \rceil+(n-i+1)-1}{2}}  \leq 2^{\lceil \frac{t}{2} \rceil} \cdot\frac{1}{\sqrt{\lceil \frac{t}{2} \rceil}},
  \]
  and 
\[\binom{\lfloor \frac{t}{2} \rfloor}{\frac{\lfloor \frac{t}{2} \rfloor+i-1}{2}} \leq  2^{\lfloor \frac{t}{2} \rfloor} \cdot\frac{1}{\sqrt{\lfloor \frac{t}{2} \rfloor}}.
 \]
The exponential terms will cancel and
\[\frac{1}{\sqrt{\lceil \frac{t}{2} \rceil}}\cdot\frac{1}{\sqrt{\lfloor \frac{t}{2} \rfloor}}
  \leq \frac{4}{t},
\]
giving our desired bound.
\end{proof}

The upper bound in Lemma~\ref{lem:simpleMaxforProductProb} is not
sufficient for~$t$ that are asymptotically less than $n^2$, so for these~$t$ we
need to give a more detailed analysis. Therefore, we more carefully
examine the binomial coefficients that are significantly smaller than the
central coefficient for small $t$. Consequently, the exponential term will not
be sufficiently canceled by the binomial coefficient for values of $t$ that are
asymptotically smaller than $n^2$. More specifically, we show that the
function of $t$ on the right hand side of Lemma~\ref{lem:sumBoundForMaxMinEnd} is
increasing in $t$ up until approximately $n^2$.

In the following lemma we consider even length walks for simplicity. The proof
for odd length walks follows analogously.

\begin{lemma}\label{lem:MaxIsAbovenSquared}
For $n\geq 20$ and any integer $ 1 \leq i \leq n$, for all $t \leq n^2/40$ we have
\begin{align*}
  &\frac{16i(n-i+1)}{(2t)^2} \frac{1}{2^{2t}} \binom{t}{\frac{t + i -1}{2}} \binom{t}{\frac{t+(n-i+1)-1}{2}} \\
  &\hspace{7.0cm}\leq 	\frac{16i(n-i+1)}{(2t + 4)^2} \frac{1}{2^{2t+4}} \binom{t+2}{\frac{t +2 + i -1}{2}} \binom{t+2}{\frac{t+2+(n-i+1)-1}{2}},
\end{align*}
where we consider walks of length $2t$ and $2t + 4$ to ensure that $(2t)/2$ and $(2t+4)/2$ have the same parity.
\end{lemma}

\begin{proof}	

Canceling like terms implies that the desired inequality is equivalent to 
  \begin{align*}
    &\frac{1}{t^2} \binom{t}{\frac{t + i -1}{2}} \binom{t}{\frac{t+(n-i+1)-1}{2}}
    \leq 	\frac{1}{(t + 2)^2} \cdot \frac{1}{16} \cdot \binom{t+2}{\frac{t +2 + i -1}{2}} \binom{t+2}{\frac{t+2+(n-i+1)-1}{2}}.
  \end{align*}
	Examining the binomial coefficients shows that
	\[
    \binom{t}{\frac{t + i -1}{2}} \frac{(t+2)(t+1)}{\left(\frac{t+1 +i}{2}\right)\left(\frac{t+3 - i}{2}\right)} = \binom{t+2}{\frac{t +2 + i -1}{2}},
  \]
	and 
  \begin{align*}
    \binom{t}{\frac{t+(n-i+1)-1}{2}}  \frac{(t+2)(t+1)}{\left(\frac{t+2 +(n-i)}{2}\right)\left(\frac{t+ 2 - (n-i)}{2}\right)}
    =  \binom{t+2}{\frac{t+2+(n-i+1)-1}{2}}.
  \end{align*}
	Using these identities, our desired inequality equals
  \begin{align*}
	\frac{1}{t^2}
    \leq 	\frac{16^{-1}}{(t + 2)^2} \frac{(t+2)(t+1)}{\left(\frac{t+1 +i}{2}\right)\left(\frac{t+3 - i}{2}\right)} \frac{(t+2)(t+1)}{\left(\frac{t+2 +(n-i)}{2}\right)\left(\frac{t+ 2 - (n-i)}{2}\right)}.
  \end{align*}
  Further cancellation of like terms and moving the denominator on each side
  into the numerator on the other side implies that our desired inequality is
  equivalent to
  \begin{align*}
    (t + 1 + i)(t + 3 - i)(t + 2 + (n-i))(t + 2 - (n-i))
    \leq 	t^2{(t+1)^2}.
  \end{align*}
	
  It is straightforward to see that \[(t + 1 + i)(t + 3 - i)\] is maximized by $i
= 1$ and
\[(t + 2 + (n-i))(t + 2 - (n-i))\] is maximized by $n-i = 0$. Furthermore, it must be true
that either $i \geq n/2$ or $n-i \geq n/2$, so we can upper bound the left hand side of
our inequality by substituting $n/2$ for $i$ or $n-i$, and setting the other terms to the
value that maximizes the product. Hence,
\begin{align*}
  (t + 1 + i)(t + 3 - i)(t + 2 + (n-i))(t + 2 - (n-i))
  \leq (t + 2)^2\left(t + 3 + \frac{n}{2}\right)\left(t + 3 - \frac{n}{2}\right).
\end{align*}
To prove our desired inequality it now suffices to show 
$(t + 2)^2\left(t + 3 + n/2\right)\left(t + 3 - n/2\right) \leq t^2(t+1)^2$,
which is equivalent to
\[
  \left(t + 3 + \frac{n}{2}\right)\left(t + 3 - \frac{n}{2}\right)  \leq t^2\left(1 - \frac{1}{t+2}\right)^2.
\]
Expanding both sides of the inequality and rearranging terms yields
\[
  6t + 9 + \frac{2t^2}{t+2}  - \left(\frac{t}{t+2}\right)^2 \leq \frac{n^2}{4}.
\]
Given that $2t^2/(t+2) \leq 2t$, it suffices to show that
$
  8t + 9 \leq n^2/4,
$
which is true when $t \leq n^2/40$ and $n \geq 20$.
\end{proof}


We can now prove the main upper bound result of this section using the
recently developed bounds for the right hand side of the expression in
\Cref{lem:sumBoundForMaxMinEnd}.

\upperProbLeftRight*

\begin{proof}
Applying Lemmas~\ref{lem:sumBoundForMaxMinEnd}
and~\ref{lem:simpleMaxforProductProb} gives
\begin{align*}
  \prob{\WW \sim \mathcal{W}^{}(1)}{
  \ww^{(t)} = n, \maP{t} = n, \miP{t} \geq 1  }
  \leq \sum_{i=1}^n 64\left(\frac{n^2}{t^3}\right),
\end{align*}
which immediately gives the upper bound $64(n/t)^3$.
Similarly, Lemmas~\ref{lem:sumBoundForMaxMinEnd} and~\ref{lem:MaxIsAbovenSquared} imply that for $t \leq n^2/40$,
\begin{align*}
  &\prob{\WW \sim \mathcal{W}^{}(1)}{
  \ww^{(t)} = n, \maP{t} = n, \miP{t} \geq 1  }\\
&\hspace{4.5cm}\leq \sum_{i=1}^n \left(\frac{16i(n-i+1)}{T^2}\right) \binom{\lfloor \frac{T}{2} \rfloor}{\frac{\lfloor \frac{T}{2} \rfloor+i-1}{2}} \frac{1}{2^{\lfloor \frac{T}{2} \rfloor}}
  \cdot\binom{\lceil \frac{T}{2} \rceil}{\frac{\lceil \frac{T}{2} \rceil+(n-i+1)-1}{2}} \frac{1}{2^{\lceil \frac{T}{2} \rceil}},
\end{align*}
where $T = n^2/40$. We then 
use Lemma~\ref{lem:simpleMaxforProductProb} and sum from $1$ to $n$ to obtain 
\begin{align*}
  \prob{\WW \sim \mathcal{W}^{}(1)}{
		\ww^{(t)} = n, \maP{t} = n, \miP{t} \geq 1  }
  &\leq \sum_{i=1}^n \frac{64n^2}{T^3}\\
  &\le \frac{64(40)^3}{n^3}\\
  &\leq \frac{e^{25}}{n^3},
\end{align*}
for all $t \leq n^2/40$.
Using the fact that $64(n/t)^3$ is a decreasing function in $t$, we have
\[
  64\left(\frac{n}{t}\right)^3 \leq \frac{e^{25}}{n^3},
\]for all $t \geq n^2 /40$, which is the desired bound.
\end{proof}


\section{Extension to Higher Dimensions}
\label{sec:higher-dimension}
Now we show how to extend our analysis to
upper and lower bound the transience class of
$d$-dimensional grids.

\DimD*

We denote by $\cube{n}$ the $d$-dimensional hypercube grid with $n^d$ vertices,
and construct it analogously to $\square{n}$. Its vertex set is $\{1,2,\dots,n\}^d
\cup \{\sink\}$ and its edges connect any pair of vertices that differ in one
coordinate. Vertices on the boundary have additional edges connecting to
$\sink$ so that every non-sink vertex has degree $2d$. We use the vector
notation $\uu = (\uu_1,\uu_2,\dots,\uu_d)$ to identify non-sink vertices.
We can decouple a walk $w$ on $\cube{n}$ into 
one-dimensional walks $w_1,w_2,\dots,w_d$, 
so that each step of a random walk on $\cube{n}$ can be
understood as choosing a random direction with probability $1/d$ and then
a step in the corresponding one-dimensional walk with probability
$1/2$.

Our bounds for the two-dimensional grid heavily relied on decoupling walks into
interleaved one-dimensional walks, and applying bounds from \Cref{sec:path} for simple symmetric
walks.
Generalizing these bounds to $d$-dimensional hypercubes
follows comparably and only requires simple extensions of our 
lemmas for the two-dimensional grids.
Therefore, we will reference the necessary lemmas from previous sections and
show the minor modifications needed to give analogous lemmas for the
$d$-dimensional grid. The upper bound proof requires several key lemmas and
is more involved, whereas extending the lower bound only requires one
simple addition to our proof in Section~\ref{sec:lower}.

\subsection{Upper Bounding the Transience Class}

Since Theorem~\ref{thm:reduction} from~\cite{ChoureV12} relies on non-sink
vertices having constant degree, our assumptions that $d$ is constant and
that all non-sink vertices have degree $2d$.
In addition to utilizing properties of one-dimensional walks,
specifically Lemma~\ref{lem:probLeftRight} proven in Section~\ref{sec:path},
the proof of our upper bound relies on four key lemmas:

\begin{itemize}
  \item Lemma~\ref{lem:swapSourceTarget} --- The source vertex can be swapped
    with a any non-sink vertex while only losing a $O(\log{n})$ approximation
    factor in the potential.
	
  \item Lemma~\ref{lem:sumBound} --- An upper bound on the sum of all vertex
    potentials by factoring the expression into one-dimensional vertex
    potentials.
	
  \item Lemma~\ref{lem:nnIsMin} --- For any vertex,
    the opposite corner vertex minimizes the potential up to a constant.
	
  \item  Lemma~\ref{lem:voltageLower} --- A lower bound on the vertex potential
    $\ppi^{}_{(n,n)}\left( \uu \right)$, for any $\uu$ in the
    top-right quadrant of $\square{n}$.
	
\end{itemize}

Now we describe how to extend each of these lemmas to constant dimensions.
These results almost immediately follow from decoupling walks
into one-dimensional walks.

\begin{lemma}
	\label{lem:swapSourceTargetHigher}
	For any pair of non-sink vertices $\uu$ and $\vv$ in $\cube{n}$, we have
	\[
    \ppi^{}_{\uu}\left(\vv\right)
  \leq \left(8\log n + 4\right) \ppi^{}_{\vv}\left(\uu\right).
	\]
\end{lemma}

\begin{proof}
This is a consequence of Rayleigh's monotonicity theorem. Fix an
underlying $n \times n$ subgraph of the hypercube with corners at the source
and sink, and set the rest of the resistors to infinity. The upper bound
for the $n \times n$ grid is an upper bound for the hypercube.
\end{proof}

Our lemma analogous to Lemma~\ref{lem:sumBound} follows
from Lemma~\ref{lem:swapSourceTargetHigher} and
decoupling walks into one-dimension.

\begin{lemma}\label{lem:sumBoundHigher}
For any non-sink vertex $\uu$ in $\cube{n}$,
\[
  \sum_{\vv \in V} \pi^{}_{\uu}(\vv) = O\left(\log{n} \prod_{i = 1}^d \uu_i \log{n}\right).
\]
\end{lemma}

\begin{proof}
Follow the proof structure of \Cref{lem:sumBound}.
\end{proof}

We can also generalize our proof of Lemma~\ref{lem:nnIsMin} to higher
dimensions, because we work with each dimension independently.

\begin{lemma}\label{lem:cornerIsMinForHigherDimension}
If $\uu$ is a non-sink vertex of $\cube{n}$ such that
$1 \le \uu_i \leq \lceil n/2 \rceil$ for all $1 \le i \le d$, then
\[
  \ppi^{}_{\uu}\left(\vv\right) \geq
  \left(\frac{1}{2d}\right)^d \ppi^{}_{\uu}\left( \left(n,n,\dots,n\right) \right).
\]
\end{lemma}

\begin{proof}
Extend the proof of \Cref{lem:nnIsMin} by reflecting walks across the
$(d-1)$-dimensional hyperplane perpendicular to the
chosen axis instead of a line.
\end{proof}

Lastly, we generalize Lemma~\ref{lem:voltageLower}, where the key
idea was to considers walks of length $\Theta(n^2)$ and show that there is a
constant fraction such that both dimensions have taken $\Theta(n^2)$ steps,
which allows us to apply Lemma~\ref{lem:probLeftRight} for each
possible walk.
To do this, we essentially
union bound Lemma~\ref{lem:chernoff} over $d$ dimensions, which shows that
$\Theta(n^2)$ walk lengths take $\Theta(n^2)$ steps in each direction with
probability at least $2^{-d}$. 

\begin{lemma}
	\label{lem:voltageLowerHigher}
  For $n \geq 10$ and $\uu \in V(\cube{n})$ such that 
	$1 \leq \uu_i \leq \lceil n/2 \rceil$ for $1 \le i \le d$, we have
	\[
    \ppi^{}_{(n,n,\dots,n)}\left( \uu \right)
	= \Omega \left(\frac{\prod_{i=1}^d\uu_i }{n^{3d - 2}} \right).
	\]
\end{lemma}

\begin{proof}
Decouple walks $\WW \in \mathcal{W}^{} ( \uu \rightarrow (n,n,\dots,n))$ into
one-dimensional walks $\WW_i \in \mathcal{W}^{\line}(\uu_i)$, and view
$\ppi^{}_{(n,n,\dots,n)}\left( \uu \right)$ as the probability that each walk
$\WW_i$ is present on $n$ at the same time before any leaves the interval $[1,n]$. If
each walk takes $t_1,t_2,\dots,t_d$ steps, respectively, then the total
number of possible interleavings of these walks is the multinomial
\[
  {t_1 + t_2 + \dots + t_d \choose t_1,t_2, \dots ,t_d}.
\]
Just as before, we can obtain the lower bound
\begin{align*}
  &\ppi^{}_{(n,n,\dots,n)}\left( \uu \right) \geq \sum_{t_1,t_2,\dots, t_d \geq 0} \frac{ {t_1 + t_2 + \dots + t_d \choose t_1,t_2, \dots ,t_d} }{d^{t_1 + t_2 + \dots + t_d}}
  \prod_{i=1}^d \frac{1}{2}
  \prob{}{\ww_i^{(t_i-1)} = n-1,
	\maP{t_i - 1} = n -1, 
	\miP{t_i-1} \geq 1 }.
\end{align*}

To apply Lemma~\ref{lem:probLeftRight} to each walk, we need each~$t_i$
to be in the interval $[n^2/c,n^2/4]$, for $c = 16d$.
Then we consider all walks of length
  \[ \frac{n^2}{8} \leq t \leq \frac{n^2}{4},\] where $t = t_1 + t_2 + \cdots + t_d$, and show that a constant
fraction of these walks satisfy $t_i \geq n^2/c$ with
$t_i$ having the correct parity. Note that we can ignore the parity
conditions by simply lower bounding the probability of all having correct
parity by $4^{-d}$. It then remains to show that all walks satisfy the
  inequality $t_i \geq n^2/c$ with constant probability. 

Consider the probability that $t_1 \geq n^2/c$.
The other dimensions follow identically. Letting each dimension take at least
$n^2/c$ steps introducing dependence, so we instead consider the probability
that $t_1 \geq n^2/c$ and condtion on $t_2,t_3,\dots,t_d \geq n^2/c$ (which
can only decrease the probability of the event $t_1 \geq n^2/c$). This is equivalent to
fixing $n^2/c$ steps in each of those directions and randomly choosing all
remaining steps with probability $1/d$ for each direction. The remaining number
of steps is then at least $dn^2/c$ by our assumption that $t \geq
n^2/8$. Therefore, the expected number of steps in the first
dimension is at least $n^2/c$, which implies $t_1 \geq n^2/c$
with probability at least $1/2$. Multiplying this probability over all
dimensions gives $t_i \geq n^2/c$ with probability at least $2^{-d}$.

Thus, there are $O(n^2)$ values of $t$ that we can decompose
into one-dimensional walks, each occurring with constant probability.
Applying Lemma~\ref{lem:probLeftRight} to each decomposition
and summing \[
  \Omega\left(\prod_{i=1}^d \frac{\uu_i}{n^{3}}\right)\] over
$O(n^2)$ possible walks proves the claim.
\end{proof}


Now we prove $\tcl(\cube{n}) = O(n^{3d-2}\log^{d+2}{n})$ using Theorem~\ref{thm:reduction}.
For any $\uu = (\uu_1, \uu_2,...,\uu_d)$ in the top-left orthant of $\cube{n}$,
it follows that
\begin{align*}
 \max_{\uu, \vv \in V\setminus\{\sink\}}
  \left( \sum_{\xx \in V}\pi^{}_{\uu}(\xx) \right)
  \pi^{}_{\uu}(\vv)^{-1} 
  &\le \max_{\uu \in V\setminus\{\sink\}} 
  \left(\sum_{\xx \in V} \ppi^{}_{\uu}\left( \xx \right)\right)
  \frac{\left(2d\right)^d}{\ppi^{}_{\uu}\left((n)^d \right)}\\
&= \max_{\uu \in V\setminus\{\sink\}}
  \left(\sum_{\xx \in V} \ppi^{}_{\uu}\left( \xx \right)\right)
  \frac{O\left(\log{n}\right)}{\ppi^{}_{(n)^d} \left( \uu\right) }\\
&= \max_{\uu \in V\setminus\{\sink\}}	
  O\left(\log{n} \prod_{i = 1}^d \uu_i \log{n}\right) 
O\left(\frac{ n^{3d - 2}\log{n}}{\prod_{i=1}^d\uu_i } \right)\\
&= O \left(n^{3d - 2}\log^{d+2}{n}\right).
\end{align*}
\noindent

\subsection{Lower Bounding the Transience Class}

Extending our lower bound to $d$-dimensional hypergrids is a simple
consequence of decoupling $d$-dimensional walks into one-dimensional walks,
because we only need to generalize the upper bound in
Lemma~\ref{lem:upperBoundMaxSum} to 
\begin{align*}
  \pi^{}_{(n)^d}\left((1)^d\right) &\leq
    \max_t{\left\{\prob{\WW \sim \mathcal{W}^{}(1)}{
      \ww^{(t)} = n, \maP{t} = n, \miP{t} \geq 1  } \right\}}\\
    & \hspace{0.42cm}\cdot d\sum_{t \geq 0} \prob{\WW \sim \mathcal{W}^{}(1)}{
      \ww^{(t)} = n, \maP{t} = n, \miP{t} \geq 1  },
	\end{align*}
by replacing the negative binomial distribution with the negative multinomial distribution.
	

\begin{fact}\label{fact:binSumInfHigherDimension}
For any nonnegative integer $t_1$, we have
\[
  \sum_{t_2,\dots, t_d \geq 0} {t_1 + t_2 + \dots + t_d \choose t_1,t_2, \dots ,t_d}\frac{1}{d^{t_1 + t_2 + \dots + t_d}} = d.
\]
\end{fact}

\begin{proof}
Consider the proof of Fact~\ref{fact:binSumInf} using the negative multinomial
distribution.
\end{proof}
Thus, we can apply Lemma~\ref{lem:upperProbLeftRight} and
Lemma~\ref{lem:upperBoundSum} to show
\begin{align*}
  \pi^{}_{(n)^d}\left((1)^d\right) &= O\left(\left(\frac{1}{n^3}\right)^{d-1}
  \frac{1}{n}\right)
  = O\left(n^{-3d + 2}\right).
\end{align*}
By Theorem~\ref{thm:reduction}, we have $\tcl(\cube{n}) = \Omega(n^{3d - 2})$.

\section*{Acknowledgments}
We thank Marcel Celaya, Wenhan Huang, Wenhao Li, Pinyan Lu, Richard Peng,
Dana Randall, and Chi Ho Yuen
for various helpful discussions. We also thank the anonymous reviewers for
their insightful comments, especially those that pointed out a major issue with
Lemma~\ref{lem:voltageLower} in a previous version of this manuscript.

\begin{appendix}
\section{Appendix for \Cref{sec:background}}
\label{sec:app_background}


In this section, we prove a relationship between voltage potentials and the
probability of a random walk escaping at the source instead of the sink.

\normalizingConstant*

\begin{proof}
By definition, we have
\[
  \ppi^{}_{\uu}\left(\vv\right) = \frac{\sum_{\WW \in \mathcal{W}^{}\left(\vv \rightarrow \uu\right)} 4^{-|\WW|}}{\sum_{\WW \in \mathcal{W}^{}(\vv\rightarrow \{\uu, \sink\})} 4^{-|\WW|}}.
\]
For any $\vv\in V(\square{n})$, let
\[
  f(\vv) = \sum_{\WW \in \mathcal{W}^{}(\vv\rightarrow \{\uu,\sink\})} 4^{-|\WW|}
\]
be the normalizing constant for $\ppi^{}_{\uu}\left(\vv\right)$.  It
follows that $f(\uu) = 1$ and $f(\sink) = 1$, because the only such walk for each
has length 0. 
For all other $\vv \in V(\square{n})\setminus\{\uu,\sink\}$, we have
\[
  f(\vv) = \frac{1}{4} \sum_{\xx\sim\vv} f(\xx).
\]
Therefore, $f(\vv)$ is a harmonic function with constant boundary values,
so $f(\vv) = 1$ for all vertices $\vv \in V(\square{n})$.
\end{proof}

We also verify that the effective resistance between $\sink$ and any internal vertex
is bounded between $\Omega(1)$ and $O(\log n)$ using a triangle
inequality for effective resistances and the fact that the effective resistance
between opposite corners
in an $n\times n$ resistor network is $\Theta(\log n)$.
This proof easily generalizes to any pair of vertices in $\square{n}$.

\begin{proposition}[{\cite{LevinPW09}}]\label{prop:cornerToCornerER}
Let $G$ be an $n\times n$ network of unit resistors.
If $u$ and $v$ are vertices at opposite corner vertices, then
$
  \log(n-1)/2 \le \er^{}\left(u,v\right) \le 2 \log n.
$
\end{proposition}

\boundedERGrid*

\begin{proof}
We first prove the lower bound
\[
  1/4 \le \er^{}(\sink, \uu).
\]
The effective resistance between $\sink$ and $\uu$ is the reciprocal of the
  total current flowing into the circuit when $\ppi^{}_{\uu}(\uu) = 1$
and $\ppi^{}_{\uu}(\sink) = 0$.
Since $\ppi^{}_{\uu}$ is a harmonic function,
  we have $\ppi^{}_{\uu}(\vv) \ge 0$ for all $\vv \in V(\square{n})$.
Moreover, $\deg{\uu}=4$, so
\[
  \er^{}(\sink, \uu) = \left(\sum_{\vv\sim\uu}
         \ppi^{}_{\uu}(\uu) - \ppi^{}_{\uu}(\vv) \right)^{-1}
  \ge \frac{1}{4}.
\]

For the upper bound, we use Rayleigh's monotonicity law,
\Cref{prop:cornerToCornerER}, and the triangle inequality for effective
resistances to show that
\[
  \er^{}(\sink, \uu) \le 2 \log n + 1,
\]
for $n$ sufficiently large.
Rayleigh's monotonicity law \cite{DoyleS84}
states that if the resistances of a circuit are increased, the effective
resistance between any two points can only increase.
The following triangle inequality for effective resistances is given in
\cite{Tetali91}:
\[
  \er^{}\left(u,v\right) \le \er\left(u,x\right) + \er\left(x,v\right).
\]

Define $H$ to be the subgraph of $\square{n}$ obtained by deleting $\sink$ and all
edges incident to $\sink$.
Let $m$ be the largest positive integer such that
$\uu_1 + i \le n$ and $\uu_2 + j \le n$ for all $0 \le i, j < m$, and
let $H(\uu)$ be the subgraph of $H$ induced by the vertex set
\[
  \{(\uu_1 + i, \uu_2 + j) : 0 \le i, j < m\}.
\]
We can view $H(\uu)$ as the largest square resistor network in $H$ such that
$\uu$ is the top-left vertex.
Let $\vv = [\uu_1 + m-1, \uu_2 + m-1]$ be the bottom-right vertex in $H(\uu)$.
Using infinite resistors to remove every edge in
$E(\square{n})\setminus E(H(\uu))$, we have
\[
  \er^{\square{n}} (\vv, \uu) \le \er^{H(\uu)} (\vv, \uu)
\]
by Rayleigh's monotonicity law.
\Cref{prop:cornerToCornerER} implies that
\[
  \er^{H(\uu)} (\vv, \uu) \le 2 \log n
\]
since $m \le n$. The vertex $\vv$ is incident to $\sink$ in $\square{n}$, so
Rayleigh's monotonicity law gives
\[
  \er^{\square{n}}(\sink, \vv) \le 1.
\]
By the triangle inequality for effective resistances, we have
\[
 \er^{}(\sink, \uu) \le
   \er^{}(\sink, \vv) + \er^{}(\vv, \uu)
   \le 2\log n + 1,
\]
which completes the proof.
\end{proof}

\section{Appendix for \Cref{sec:grid}}
\label{sec:app_grid}


We use the random walk interpretation of voltage to prove
\Cref{lem:nnIsMin}.
The key idea is that the voltage on the boundary opposite of $\uu$ along any
axis is less by a constant factor. This projection can be iterated along an
axis in each dimension.


\nnIsMin*

\begin{proof}
We use \Cref{lem:normalizingConstant} to decompose $\ppi^{}_{\uu}(\vv)$
as a sum of probabilities of walks, and
then construct maps for all $1 \le \vv_1, \vv_2 \le n$ to show
\begin{align*}
  \ppi^{}_{\uu} \left( \left(\vv_1, \vv_2 \right)\right)
    & \geq 
  \max\left\{
    \frac{1}{4} \ppi^{}_{\uu}\left( \left(n, \vv_2 \right)\right),~
  \frac{1}{4} \ppi^{}_{\uu}\left( \left(\vv_1, n \right)\right)
  \right\}.
\end{align*}

\noindent
We begin by considering the first dimension:
  \[\ppi^{}_{\uu} \left( (\vv_1, \vv_2) \right) \ge 
  \frac{\ppi^{}_{\uu} \left((n, \vv_2)\right)}{4}.
  \]
Let $\ell_{\text{hor}}$ be the horizontal line of reflection passing through
$(\lceil(\vv_1 + n)/2\rceil, 1)$ and $(\lceil(\vv_1 + n)/2\rceil, n)$
in $\mathbb{Z}^2$,
and let $\uu^*$ be the reflection of $\uu$ over $\ell_{\text{hor}}$.
Note that $\uu^*$ may be outside of the $n \times n$ grid.
Next, define the map
\[
  f : \mathcal{W}^{}((n,\vv_2) \rightarrow \uu) \rightarrow
      \mathcal{W}^{}((\vv_1,\vv_2) \rightarrow \uu)
\]
as follows.
For any walk $w \in \mathcal{W}^{}((n,\vv_2) \rightarrow \uu)$:
\begin{enumerate}
  \item Start the walk $f(w)$ at $(\vv_1,\vv_2)$, and if $n-\vv_1$ is odd
    move to $(\vv_1+1,\vv_2)$.
  \item Perform $w$ but make opposite vertical moves before the walk hits
    $\ell_{\text{hor}}$, so that the partial walk is a reflection over
    $\ell_{\text{hor}}$.
  \item After hitting $\ell_{\text{hor}}$ for the first time, continue
    performing $w$, but now use the original vertical moves.
  \item Terminate this walk when it first reaches $\uu$.
\end{enumerate}

\noindent
Denote the preimage of a walk
$w' \in \mathcal{W}^{}((\vv_1,\vv_2) \rightarrow \uu)$ under $f$ to be
\[
  f^{-1}\left(w'\right)
  = \left\{w \in \mathcal{W}^{}((n,\vv_2) \rightarrow \uu) : f(w) = w'\right\}.
\]
We claim that for any
$w' \in \mathcal{W}^{\square{n}}((\vv_1,\vv_2) \rightarrow \uu)$, 
\[
  \frac{1}{4} \sum_{w \in f^{-1}(w')} 4^{-|w|} \le 4^{-|w'|}.
\]
If $f^{-1}(w') = \emptyset$ the claim is true, so assume
$f^{-1}(w') \ne \emptyset$. We analyze two cases.
If $w'$ hits $\ell_{\text{hor}}$, then $f^{-1}(w')$ contains exactly one walk
$w$ of length $|w'|$ or $|w'|-1$. If $w'$ does not hit $\ell_{\text{hor}}$,
then
\begin{align*}
  f^{-1}(w') = \{w \in \mathcal{W}^{}((n,\vv_2) \rightarrow \uu) :
  \text{$w$ is a reflection of
  $w'$ over $\ell_{\text{hor}}$
  before $w$ hits $\uu^*$}\}.
\end{align*}
It follows that any walk $w \in f^{-1}(w')$ can be split into $w=w_1 w_2$,
where $w_1$ is the unique walk from $(n,\vv_2)$ to $\uu^*$ that is a reflection of 
$w'$, and $w_2$ is a walk from $\uu^*$ to $\uu$ that avoids $\sink$
and hits $\uu$ exactly once upon termination. Clearly $w_1$ has length
$|w'|$ or $|w'|-1$, and the set of admissible $w_2$ is
$\mathcal{W}^{}(\uu^* \rightarrow \uu)$. Therefore,
\begin{align*}
  \frac{1}{4} \sum_{w \in f^{-1}(w')} 4^{-|w|} &=
  4^{-|w_1|-1} \sum_{w_2 \in \mathcal{W}^{}(\uu^* \rightarrow \uu)}
  4^{-|w_2|}\\
  &= 4^{-|w_1|-1} \ppi^{}_{\uu}\left(\uu^*\right)\\
  &\le 4^{-|w'|},
\end{align*}
since $\ppi^{}_{\uu}(\uu^*)$ is an escape probability.
Summing over all
$w' \in \mathcal{W}^{}((\vv_1,\vv_2) \rightarrow \uu)$,
it follows from \Cref{lem:normalizingConstant} and the previous inequality that
\begin{align*}
  \ppi^{}_{\uu}\left((\vv_1,\vv_2)\right)
  &= \sum_{w' \in \mathcal{W}^{}((\vv_1,\vv_2) \rightarrow \uu)} 4^{-|w'|}\\
  &\ge \sum_{w' \in \mathcal{W}^{}((\vv_1,\vv_2) \rightarrow \uu)}
          \frac{1}{4} \sum_{w \in f^{-1}(w')} 4^{-|w|}\\
  &\ge \frac{1}{4}\ppi^{}_{\uu}\left((n,\vv_2)\right),
\end{align*}
because every $w \in \mathcal{W}^{}((n,\vv_2) \rightarrow \uu)$ is the
preimage of a $w' \in \mathcal{W}^{}((\vv_1,\vv_2) \rightarrow \uu)$.

Similarly, we can show that $\ppi^{}_{\uu}\left((\vv_1, \vv_2)\right)
\geq \ppi^{}_{\uu}\left((\vv_1,n)\right) / 4$ for all $1 \le \vv_1 \le n$
by reflecting walks over the vertical line from
$(1, \lceil (n+ \vv_2)/2 \rceil)$ to $(n, \lceil (n + \vv_2)/2 \rceil)$.
Combining inequalities proves the claim.
\end{proof}


Lastly, we give a constant lower bound for the probability of an $n$-step
simple symmetric walk being sufficiently close to its starting position by using
the recursive definition of binomial coefficients and a Chernoff bound for
symmetric random variables.

\chernoff*



\begin{proof}
First observe that for $n \ge 10$, we have
\[
	\frac{1}{2^n} \sum_{\substack{k = \left\lceil \frac{n}{4} \right\rceil \\ k 		\text{ odd}}}^{\left\lfloor \frac{3n}{4} \right\rfloor} \binom{n}{k}
	\ge
	\frac{1}{2^n} \sum_{k \in \left(\frac{n-1}{4}, \frac{3(n-1)}{4} \right) } \binom{n - 1}{k}
\]
and
\[
	\frac{1}{2^n} \sum_{\substack{k = \left\lceil \frac{n}{4} \right\rceil \\ k 
	\text{ even}}}^{\left\lfloor \frac{3n}{4} \right\rfloor} \binom{n}{k}
	\ge
	\frac{1}{2^n} \sum_{k \in \left(\frac{n-1}{4}, \frac{3(n-1)}{4} \right) } \binom{n - 1}{k}.
\]
To see this, use the parity restriction and expand the summands as
\[
  \binom{n}{k} = \binom{n-1}{k-1} + \binom{n-1}{k}.
\]
Let $X_1, X_2, \dots, X_{n-1}$ be independent Bernoulli random
variables such that $\prob{}{X_i = 0} = 1/2$ and
$\prob{}{X_i = 1} = 1/2$.
Let $S_{n-1} = X_1 + X_2 + \dots + X_{n-1}$ and $\mu =
E[S_{n-1}] = (n-1)/2$.
Using a Chernoff bound, we have
\begin{align*}
  \frac{1}{2^n} \sum_{k \in \left(\frac{n-1}{4}, \frac{3(n-1)}{4} \right) } \binom{n - 1}{k}
  &= \frac{1}{2}\left(1 -  \prob{}{\left|S_{n-1} - \mu \right| \ge \frac{1}{2} \mu}\right)\\
  &\ge \frac{1}{2} - e^{-(n-1)/24}\\
  &\ge \frac{2}{5},
\end{align*}
for $n \ge 60$. Checking the remaining cases numerically when $10 \le n < 60$
proves the claim.
\end{proof}

\end{appendix}

{
\bibliographystyle{alpha}
\bibliography{ref}
}

\end{document}